\documentclass[12pt,draftclsnofoot,onecolumn]{IEEEtran}

\usepackage{bbm}
\usepackage{url}
\usepackage{calc}
\usepackage{cite}
\usepackage{soul}
\usepackage{tikz}
\usepackage[normalem]{ulem}
\usepackage{color}
\usepackage{cuted}
\usepackage{float}
\usepackage{amsthm}
\usepackage{lipsum}
\usepackage{amsmath}
\usepackage{amssymb}
\usepackage{caption}
\usepackage{graphicx}
\usepackage{pgfplots}
\usepackage{stmaryrd}
\usepackage{thmtools}
\usepackage{algorithm}
\usepackage{clipboard}
\usepackage{tikzscale}
\usepackage{mathtools}
\usepackage{microtype}
\usepackage{algorithmic}
\usepackage{thm-restate}

\pgfplotsset{compat=newest}
\usetikzlibrary{arrows}
\usetikzlibrary{patterns}
\usetikzlibrary{plotmarks}
\usetikzlibrary{arrows.meta}
\usetikzlibrary{positioning}
\usepgfplotslibrary{patchplots}
\usetikzlibrary{external}
\allowdisplaybreaks

\theoremstyle{plain} \newtheorem{theorem}{Theorem}
\theoremstyle{plain} 
\theoremstyle{plain} \newtheorem{corollary}[theorem]{Corollary}
\theoremstyle{plain} \newtheorem{proposition}{Proposition}
\theoremstyle{definition} \newtheorem{definition}{Definition}
\theoremstyle{plain} 
\theoremstyle{remark} \newtheorem{remark}{Remark}

\newlength{\mygraphwidth}
\newlength{\mygraphheight}
\IEEEoverridecommandlockouts

\begin{document}
	\title{Precoder Design and Statistical Power Allocation for MIMO-NOMA via User-Assisted Simultaneous Diagonalization}
	\author{\IEEEauthorblockN{Aravindh Krishnamoorthy\rlap{\textsuperscript{\IEEEauthorrefmark{1}\IEEEauthorrefmark{2}}},\,\,\, Zhiguo Ding\rlap{\textsuperscript{\IEEEauthorrefmark{3}}},\, and Robert Schober\rlap{\textsuperscript{\IEEEauthorrefmark{1}}}}\\
	\IEEEauthorblockA{\small \IEEEauthorrefmark{1}Friedrich-Alexander-Universit\"{a}t Erlangen-N\"{u}rnberg,  \IEEEauthorrefmark{2}Fraunhofer Institute for Integrated Circuits (IIS),\\\IEEEauthorrefmark{3}The  University  of  Manchester \vspace{-1.5cm}} \thanks{This paper was presented in part at the IEEE Int. Conf. Commun. (ICC) 2019 \cite{Krishnamoorthy2019} and the 24th Intl. ITG Workshop on Smart Antennas \cite{Krishnamoorthy2019a}. Computer programs for the most important analytical results in this paper can be downloaded from \protect\url{https://gitlab.com/aravindh.krishnamoorthy/mimo-noma}.}}
	\maketitle
	
	\begin{abstract}
		In this paper, we investigate the downlink precoder design for two-user power-domain multiple-input multiple-output (MIMO) non-orthogonal multiple access (NOMA). We propose a novel user-assisted (UA) simultaneous diagonalization (SD) based MIMO-NOMA scheme that achieves SD of the MIMO channels of both users through a combination of precoder design and low-complexity self-interference cancellation at the users, thereby considerably lowering the overall decoding complexity compared to joint decoding. The achievable ergodic user rates of the proposed scheme are analyzed for Rayleigh fading channels based on a finite-size random matrix theory framework, which is further exploited to develop a statistical power allocation algorithm. Simulation and numerical results show that the proposed UA-SD MIMO-NOMA scheme significantly outperforms orthogonal multiple access and a benchmark precoder design performing SD via generalized singular value decomposition in terms of the achievable ergodic rate region for most user rates. The ergodic rate region is further enhanced by a hybrid scheme which performs time sharing between the proposed UA-SD MIMO-NOMA scheme and single-user  MIMO.
	\end{abstract}
	
	\section{Introduction}
	Non-orthogonal multiple access (NOMA), proposed in \cite{Saito2013}, aims to improve the downlink spectral efficiency, connectivity, and user fairness\footnote{NOMA can improve user fairness if  appropriate power allocation is applied \cite{Saito2013}.} of wireless communication systems, by exploiting superposition coding at the transmitter and successive interference cancellation (SIC) at the receivers, compared to orthogonal multiple access (OMA) \cite{Ding2017, Islam2017}. Although the authors of \cite{Saito2013} primarily considered NOMA employing power-domain superposition coding, other forms of NOMA based on code division multiple access (CDMA) have been proposed as well \cite{Dai2018, Ding2017, Islam2017}. Nevertheless, in this paper, we focus on power-domain NOMA owing to its simplicity and compatibility with the existing 4th generation (4G) networks.
	
	While most of the early works on NOMA considered single-antenna transmitters \cite{Saito2013}, \cite{Ding2014}, the extension of NOMA to multiple-input multiple-output (MIMO) systems is of interest as it has the potential of combining the benefits of the multiple spatial streams facilitated by MIMO with the increased spectral efficiency and user fairness enabled by NOMA. Various works have shown that power-domain MIMO-NOMA enables significantly higher data rates compared to MIMO-OMA \cite{Liu2016,Zeng2017}. However, a careful MIMO precoder design was found to be crucial for realizing the potential performance gains \cite{Ding2016}. To this end, multiple MIMO-NOMA precoder designs have been reported. In particular, a model for quasi-degradation was proposed for multiple-input single-output (MISO) channels in \cite{Chen2016,Chen2017a}, and exploited for the development of a precoder for MISO channels. Extensions to MIMO channels were reported in \cite{Ding2016,Ali2017,Zeng2017a} based on user signal alignment. Optimal power allocation to further improve the performance of various precoding schemes was considered in \cite{Hanif2019,Liu2019,Wang2019,Xiao2019,Shao2019,Ding2019}.
	
	However, most existing precoding schemes entail a high decoding complexity at the receivers owing to the need for self- and inter-user-interference cancellation. In order to reduce the decoding complexity at the users, the authors of \cite{Ma2016,Chen2019} proposed a two-user MIMO-NOMA scheme employing simultaneous diagonalization (SD) of the MIMO channels of the users via generalized singular value decomposition (GSVD). The resulting precoding scheme decomposes the MIMO-NOMA channels of the users into multiple single-input single-output (SISO)-NOMA channels, thereby enabling low-complexity decoding. In addition, the authors of \cite{Chen2019} analyzed the achievable ergodic user rates of GSVD-based precoding for asymptotically large numbers of antennas via random matrix theory (RMT) for Rayleigh fading channels with equal power allocation (EPA).
	
	Although GSVD-based precoding enables SD, it also suffers from several shortcomings. Importantly, GSVD-based precoding has poor performance if the sum of the numbers of receiver antennas approaches the number of antennas at the base station (BS). Furthermore, GSVD requires the inversion of the MIMO channels of both users, which also degrades performance.
	
	In this paper, we propose a novel user-assisted (UA)-SD MIMO-NOMA scheme that simultaneously diagonalizes the MIMO channels of both users through a combination of precoder design and low-complexity self-interference cancellation at the users, and overcomes the shortcomings of GSVD-based precoding. Furthermore, we develop a finite-size RMT based framework in order to analyze the achievable ergodic rate of the proposed MIMO-NOMA precoding scheme in Rayleigh fading for a finite number of antennas, where both EPA and unequal power allocation (UPA) are considered. Moreover, we present large finite-dimension approximations for the achievable ergodic rate.
	
	This paper builds upon the conference versions in \cite{Krishnamoorthy2019} and \cite{Krishnamoorthy2019a}, which introduced UA-SD MIMO-NOMA for the case where the number of receiver antennas is less than the number of BS antennas, and EPA with fixed \cite{Krishnamoorthy2019} and flexible \cite{Krishnamoorthy2019a} partitioning of the power between the users. In this paper, we generalize the proposed MIMO-NOMA scheme to all possible antenna configurations and UPA. Furthermore, we extend the analysis and simulation results, and provide proofs for the derived analytical results. The main contributions of this paper can be summarized as follows.
	\begin{itemize}
		\item We propose downlink two-user power-domain UA-SD MIMO-NOMA precoding and decoding schemes which decompose the MIMO-NOMA channels of the users into multiple SISO-NOMA channels, assuming low-complexity self-interference cancellation at one of the users.
		\item We develop a finite-size RMT framework to evaluate the achievable ergodic user rates of the proposed MIMO-NOMA scheme in Rayleigh fading for EPA and UPA, respectively. Furthermore, we develop a long-term power allocation scheme requiring knowledge of only the channel statistics.
		\item Lastly, we exploit the developed finite-size RMT framework to obtain the ergodic achievable rate region of the proposed MIMO-NOMA scheme. A comparison with GSVD-based precoding and time division multiple access (TDMA) based MIMO-OMA reveals significant performance gains for most user rates.
	\end{itemize}
	
	The remainder of this paper is organized as follows. In Section \ref{sec:prelim}, we present the MIMO-NOMA system model and briefly review the existing SD schemes. In Section \ref{sec:proposed}, we present the proposed UA-SD MIMO-NOMA precoding and decoding schemes. In Section \ref{sec:performance}, we develop finite-size RMT based analysis frameworks for EPA and UPA, respectively, and a statistical power allocation algorithm. In Section \ref{sec:sim}, we present the simulation results, and Section \ref{sec:con} concludes the paper.
	
	\emph{Notation:}
	Boldface capital letters $\boldsymbol{X}$ and lower case letters $\boldsymbol{x}$ denote matrices and vectors, respectively. $\boldsymbol{X}^\mathrm{T},$ $\boldsymbol{X}^\mathrm{H},$  $\boldsymbol{X}^+,$ $\boldsymbol{X}^\frac{1}{2},$ $\mathrm{tr}\left(\boldsymbol{X}\right),$ and $\mathrm{det}\left(\boldsymbol{X}\right)$ denote the transpose, Hermitian transpose, Moore-Penrose pseudoinverse, symmetric square root, trace, and determinant of matrix $\boldsymbol{X},$ respectively. Furthermore, $\mathrm{col}\big(\boldsymbol{X}\big)$ and $\mathrm{null}\left(\boldsymbol{X}\right)$ denote the column space and null space of matrix $\boldsymbol{X},$ respectively. $\mathbb{C}^{m\times n}$ and $\mathbb{R}^{m\times n}$ denote the set of all $m\times n$ matrices with complex-valued and real-valued entries, respectively. The $(i,j)$-th entry of a matrix $\boldsymbol{X}$ is denoted by $[\boldsymbol{X}]_{ij}.$ Moreover, $\mathrm{diag}(d_1, \dots, d_n)$ denotes a diagonal matrix with diagonal elements $\{d_1, \dots, d_n\}.$ $\boldsymbol{I}_n$ denotes the $n\times n$ identity matrix, and $\boldsymbol{0}$ denotes an all-zeros matrix of appropriate dimension. The circularly symmetric complex Gaussian (CSCG) distribution with mean $\boldsymbol{\mu}$ and covariance matrix $\boldsymbol{\Sigma}$ is denoted by $\mathcal{CN}(\boldsymbol{\mu},\boldsymbol{\Sigma})$ and the matrix-variate Wishart distribution with parameters $p$ and $q$ and covariance matrix $\boldsymbol{\Sigma}$ is denoted by $\mathcal{CW}_p(q,\boldsymbol{\Sigma})$ \cite[Def. 3.2.1]{Gupta1999}. $\sim$ stands for ``distributed as''. $\mathcal{O}\mkern-\medmuskip\left(\cdot\right)$ and $\mathrm{E}\left[\cdot\right]$ denote the big-O notation and statistical expectation, respectively. For a tuple $\boldsymbol{t}$ containing 2-tuples as elements, i.e., $\boldsymbol{t} = ((x_1,y_1),(x_2,y_2),\dots,(x_n,y_n)), \pi_1(\boldsymbol{t}) = (x_1, x_2, \dots, x_n)$ denotes the first projection, and $\pi_2(\boldsymbol{t}) = (y_1, y_2, \dots, y_n)$ denotes the second projection. 

	\section{Preliminaries}
	\label{sec:prelim}
	In this section, we present the downlink two-user MIMO-NOMA system model and briefly review existing spatial MIMO-OMA and MIMO-NOMA precoding schemes that perform SD.
	
	\subsection{System Model}
	\label{sec:systemmodel}
	\label{ssec:systemmodel}	
	We consider a two-user\footnote{We restrict the number of paired users to two for tractability and to limit the overall decoding complexity at the receivers as pairing $K > 1$ users necessitates $K(K-1)/2$ successive interference cancellation stages. For $K>2$ users, a hybrid approach, such as in \cite[Sec. V-B]{Chen2019}, can be employed where the users are divided into groups of two users and each group is allocated orthogonal resources. Within each two-user group, the proposed MIMO-NOMA scheme can be applied.} downlink transmission, where the BS is equipped with $N$ antennas, and the users have $M_1$ and $M_2$ antennas, respectively. Furthermore, we assume that  user 1 is located farther away from the BS and experiences a higher path loss compared to user 2.
	
	The MIMO channel matrix between the BS and user $k,$ $\boldsymbol{H}_k, k=1,2,$ is modeled as
	\begin{equation}
		\frac{1}{\sqrt{\mathstrut \Pi_k}} \boldsymbol{H}_k, \label{eqn:plm}
	\end{equation}
	where matrix $\boldsymbol{H}_k \in \mathbb{C}^{M_k\times N}$ captures the small-scale fading effects, and its elements $[\boldsymbol{H}_k]_{ij} \sim \mathcal{CN}(0,1), i=1,\dots,M_k, j=1,\dots,N, k=1,2,$ are statistically independent\footnote{We assume statistically independent channel gains to facilitate the analysis of the ergodic achievable rate and UPA in Section \ref{sec:performance}. The UA-SD MIMO-NOMA precoding and decoding matrices proposed in Section \ref{sec:proposed} are also applicable if the channel gains are statistically dependent.} for all $i,j,k.$ Hence, matrices $\boldsymbol{H}_k, k=1,2,$ have full column or row rank with probability one. The scalar $\Pi_k > 0, k=1,2,$ models the path loss between the BS and user $k,$ where $\Pi_1 > \Pi_2.$
	
	Next, let $L = \mathrm{min}\left\{M_1+M_2, N\right\}$ denote the symbol vector length, and let $\boldsymbol{s}_1 = [s_{1,1},\dots,s_{1,L}]^\mathrm{T}$ $\in \mathbb{C}^{L\times 1}$ and $\boldsymbol{s}_2 = [s_{2,1},\dots,s_{2,L}]^\mathrm{T} \in \mathbb{C}^{L\times 1}$ denote the symbol vectors intended for the first and the second users, respectively. We assume that the $s_{k,l} \sim \mathcal{CN}(0,1), k=1,2, l=1,\dots,L,$ are statistically independent for all $k,l.$ We construct the MIMO-NOMA symbol vector $\boldsymbol{s} = [s_1, \dots, s_L]^\mathrm{T}$ as follows:
	\begin{align}
		\boldsymbol{s} = {\mathrm{diag}\left(\sqrt{p_{1,1}},\dots,\sqrt{p_{1,L}}\right)}\boldsymbol{s}_1 + {\mathrm{diag}\left(\sqrt{p_{2,1}},\dots,\sqrt{p_{2,L}}\right)}\boldsymbol{s}_2, \label{eqn:s}
	\end{align}
	where $\mathrm{E}\left[\boldsymbol{s}_1 \boldsymbol{s}_1^\mathrm{H}\right] = \mathrm{E}\left[\boldsymbol{s}_2 \boldsymbol{s}_2^\mathrm{H}\right] = \boldsymbol{I}_L,$ and $p_{k,l} \geq 0, k=1,2, l=1,\dots,L,$ is the transmit power allocated to the $l$-th symbol of user $k.$ The MIMO-NOMA symbol vector is precoded using a linear precoder matrix $\boldsymbol{P} \in \mathbb{C}^{N\times L},$ resulting in the transmit signal $\boldsymbol{x} = \boldsymbol{P} \boldsymbol{s}.$  The corresponding average transmit power, $P_T,$ is given by
	\begin{align}
	P_T = \mathrm{E}\left[\mathrm{tr}\left(\boldsymbol{P} \mathrm{diag}\left(p_{1,1} + p_{2,1}, \dots, p_{1,L} + p_{2,L}\right) \boldsymbol{P}^\mathrm{H}\right)\right]. \label{eqn:eppleq1}
	\end{align}
	
	At user $k,$ $k=1,2,$ the received signal, $\hat{\boldsymbol{y}}_k \in \mathbb{C}^{M_k\times 1},$ is given by
	\begin{align}
		\hat{\boldsymbol{y}}_k &= \frac{1}{\sqrt{\Pi_k}}\boldsymbol{H}_k \boldsymbol{x} + \hat{\boldsymbol{z}}_k = \frac{1}{\sqrt{\Pi_k}} \boldsymbol{H}_k \boldsymbol{P} \boldsymbol{s} + \hat{\boldsymbol{z}}_k,
	\end{align} 
	where $\hat{\boldsymbol{z}}_k \sim \mathcal{CN}(\boldsymbol{0},\sigma^2 \boldsymbol{I}_{M_k})$ denotes the additive white Gaussian noise (AWGN) vector. Furthermore, at user $k,$ $\hat{\boldsymbol{y}}_k$ is processed by a unitary detection matrix $\boldsymbol{Q}_k \in \mathbb{C}^{M_k\times M_k}$ leading to
	\begin{equation}
		\boldsymbol{y}_k = \boldsymbol{Q}_k \hat{\boldsymbol{y}}_k = \frac{1}{\sqrt{\Pi_k}} \boldsymbol{Q}_k\boldsymbol{H}_k \boldsymbol{P} \boldsymbol{s} + \boldsymbol{z}_k, \label{eqn:yk}
	\end{equation}
	where $\boldsymbol{z}_k = \boldsymbol{Q}_k\hat{\boldsymbol{z}}_k \sim \mathcal{CN}(\boldsymbol{0},\sigma^2 \boldsymbol{I}_{M_k}).$ $\boldsymbol{y}_k$ is subsequently used for detection. Furthermore, in order to be able to unveil the maximum performance of UA-SD MIMO-NOMA, for computation of the precoding and detection matrices, $\boldsymbol{P}$ and $\boldsymbol{Q}_k,k=1,2,$ perfect channel state information knowledge at the BS and the users is assumed.

	\subsection{Existing Simultaneous Diagonalization Schemes}
	\label{sec:existing}
	In this section, we briefly review existing SD precoding schemes that are based on joint zero forcing (JZF) \cite{Wiesel2008}, block diagonalization (BD) \cite{Spencer2004}, and GSVD-based MIMO-NOMA \cite{Chen2019}. Furthermore, exploiting finite-size RMT, we derive a simplified expression for the average BS transmit power in (\ref{eqn:eppleq1}) for GSVD-based MIMO-NOMA, and use this result to illustrate some of the shortcomings of this precoding scheme.
	
	\subsubsection{JZF Based Precoding \cite{Wiesel2008}}
	\label{sec:zfp}
	SD of the MIMO channels of the users can be accomplished via JZF if sufficient degrees of freedom (DoFs) are available at the BS, i.e., if $M_1+M_2 \leq N.$ In this case, from \cite{Wiesel2008}, the precoding and detection matrices are given by $\boldsymbol{Q}_1 = \boldsymbol{I}_{M_1}, \boldsymbol{Q}_2 = \boldsymbol{I}_{M_2},$ and $\boldsymbol{P} = \boldsymbol{H}^+,$ where $\boldsymbol{H} = \begin{bmatrix}\boldsymbol{H}_1^\mathrm{T} & \boldsymbol{H}_2^\mathrm{T}\end{bmatrix}^\mathrm{T}.$ Using the above precoder and detection matrices, $\boldsymbol{Q}_1 \boldsymbol{H}_1 \boldsymbol{P} = \begin{bmatrix}\boldsymbol{I}_{M_1} & \boldsymbol{0}\end{bmatrix}$ and $\boldsymbol{Q}_2 \boldsymbol{H}_2 \boldsymbol{P} = \begin{bmatrix}\boldsymbol{0} & \boldsymbol{I}_{M_2}\end{bmatrix}.$ 

	\subsubsection{BD Based Precoding \cite{Spencer2004}}
	\label{sec:bd}
	SD of the MIMO channels of the users can be performed based on BD if sufficient DoFs are available at the BS, i.e., if $M_1+M_2 \leq N.$ Let a basis of $\mathrm{null}\left(\boldsymbol{H}_1\right)$ and $\mathrm{null}\left(\boldsymbol{H}_2\right)$ be contained in $\bar{\boldsymbol{H}}_1 \in \mathbb{C}^{L\times M_2}$ and $\bar{\boldsymbol{H}}_2 \in \mathbb{C}^{L\times M_1}$ with dimensions\footnote{For $M_1+M_2 \leq N,$ bases $\bar{\boldsymbol{H}}_1$ and $\bar{\boldsymbol{H}}_2$ with dimensions $M_2$ and $M_1,$ respectively, can always be found.} $M_2$ and $M_1,$ respectively. For BD, from \cite{Spencer2004}, the precoding and detection matrices are given by $\boldsymbol{Q}_1 = \boldsymbol{U}_1^\mathrm{H},$ $\boldsymbol{Q}_2 = \boldsymbol{U}_2^\mathrm{H},$ and $\boldsymbol{P} = \begin{bmatrix}\frac{1}{\sqrt{M_1}}\bar{\boldsymbol{H}}_2\boldsymbol{V}_1 & \frac{1}{\sqrt{{M_2}}}\bar{\boldsymbol{H}}_1\boldsymbol{V}_2\end{bmatrix},$	where $\boldsymbol{U}_1 \in \mathbb{C}^{M_1\times M_1}, \boldsymbol{U}_2 \in \mathbb{C}^{M_2\times M_2}, \boldsymbol{V}_1 \in \mathbb{C}^{N\times M_1},$ and $\boldsymbol{V}_2 \in \mathbb{C}^{N\times M_2}$ are unitary matrices obtained from the SVDs $\frac{1}{\sqrt{{M_1}}}\boldsymbol{H}_1\bar{\boldsymbol{H}}_2 = \boldsymbol{U}_1 \boldsymbol{\Sigma}_1 \boldsymbol{V}_1$ and $\frac{1}{\sqrt{M_2}}\boldsymbol{H}_2\bar{\boldsymbol{H}}_1 = \boldsymbol{U}_2 \boldsymbol{\Sigma}_2 \boldsymbol{V}_2,$ and $\boldsymbol{\Sigma}_1$ and $\boldsymbol{\Sigma}_2$ contain the singular values of $\frac{1}{\sqrt{M_1}}\boldsymbol{H}_1\bar{\boldsymbol{H}}_2$ and $\frac{1}{\sqrt{M_2}}\boldsymbol{H}_2\bar{\boldsymbol{H}}_1$ on their main diagonals, respectively. Furthermore, as $\bar{\boldsymbol{H}}_2\boldsymbol{V}_1$ and $\bar{\boldsymbol{H}}_1\boldsymbol{V}_2$ are unitary matrices, (\ref{eqn:eppleq1}) simplifies to 
	\begin{align}
		P_T = \sum_{k = 1}^{2} \frac{1}{M_k} \sum_{l = 1}^{L} p_{k,l}. \label{eqn:bdpt}
	\end{align}
	Using the above precoder and detection matrices, we have $\boldsymbol{Q}_1 \boldsymbol{H}_1 \boldsymbol{P} = \boldsymbol{\Sigma}_1$ and $\boldsymbol{Q}_2 \boldsymbol{H}_2 \boldsymbol{P} = \boldsymbol{\Sigma}_2.$ 
	
	JZF and BD based precoding may be interpreted as spatial MIMO-OMA schemes and do not require SIC. Next, we review GSVD-based MIMO-NOMA which can simultaneously diagonalize the MIMO channels of the users also for $M_1 + M_2 > N,$ but requires SIC at one of the users.
	
	\subsubsection{GSVD Based Precoding \cite{Ma2016},\cite{Chen2019}}
	\label{sec:gsvd}
	GSVD \cite{VanLoan1976,Paige1981} is a matrix decomposition technique that simultaneously diagonalizes two matrices having equal numbers of columns. For the problem at hand, channel matrices $\boldsymbol{H}_1$ and $\boldsymbol{H}_2$ may be simultaneously diagonalized as follows:
	\begin{align}
		\boldsymbol{Q}_1 \boldsymbol{H}_1 \boldsymbol{Z} = \boldsymbol{C},\qquad  \boldsymbol{Q}_2 \boldsymbol{H}_2 \boldsymbol{Z} = \boldsymbol{S}, \label{eqn:gsvd}
	\end{align}
	where $\boldsymbol{Z} \in \mathbb{C}^{N\times L}$ is a full matrix, $\boldsymbol{C} \in \mathbb{R}^{M_1\times L}$ and $\boldsymbol{S} \in \mathbb{R}^{M_2\times L}$ are given by 
	\begin{align}
	\scalebox{1}{\mbox{\ensuremath{\displaystyle \boldsymbol{C} = \begin{bmatrix}
		\boldsymbol{0}	& \boldsymbol{C}_1 & \boldsymbol{0}\\
		\boldsymbol{0}	& \boldsymbol{0}	& \boldsymbol{I}_{\bar{M}_1}
	\end{bmatrix},\qquad 	\boldsymbol{S} = \begin{bmatrix}
		\boldsymbol{I}_{\bar{M}_2}	& \boldsymbol{0} & \boldsymbol{0}\\
		\boldsymbol{0}	& \boldsymbol{S}_1	& \boldsymbol{0}
	\end{bmatrix},}}} \label{eqn:cs}
	\end{align}
	such that $\boldsymbol{C}^\mathrm{H}\boldsymbol{C} + \boldsymbol{S}^\mathrm{H}\boldsymbol{S} = \boldsymbol{I}_L,$ and $\boldsymbol{Q}_1 \in \mathbb{C}^{M_1\times M_1}$ and $\boldsymbol{Q}_2 \in \mathbb{C}^{M_2\times M_2}$ are unitary matrices. Furthermore, $\bar{M}_1 = \mathrm{min}\left\{M_1,\mathrm{max}\left\{0,N-M_2\right\}\right\}, \bar{M}_2 = \mathrm{min}\left\{M_2,\mathrm{max}\left\{0,N-M_1\right\}\right\},$ and $M = N-\bar{M}_1-\bar{M}_2.$ Moreover, $\boldsymbol{C}_1 \in \mathbb{R}^{M\times M}$ and $\boldsymbol{S}_1 \in \mathbb{R}^{M\times M}$ are diagonal matrices such that $\boldsymbol{C}_1 \boldsymbol{S}_1^{-1}$ contains the generalized singular values (GSVs) of $\boldsymbol{H}_1$ and $\boldsymbol{H}_2$ on its main diagonal.
	
	Hence, by choosing the precoder matrix as $\boldsymbol{P} =\boldsymbol{Z}$ and the detection matrices for users 1 and 2 as $\boldsymbol{Q}_1$ and $\boldsymbol{Q}_2,$ respectively, from (\ref{eqn:gsvd}), the MIMO-NOMA channels of both users are simultaneously diagonalized into SISO-NOMA channels \cite{Chen2019} with $\boldsymbol{C}$ and $\boldsymbol{S}$ as the equivalent MIMO channels for users 1 and 2, respectively.
	
	
	We use GSVD-based precoding as a baseline for comparison with the proposed scheme. GSVD-based precoding \cite{Chen2019} uses identical powers for all $s_l, l = 1,\dots,L,$ i.e., $p_{1,l} + p_{2,l} = P, P \geq 0,\,\forall\,l.$ Hence, in this case, (\ref{eqn:eppleq1}) simplifies to the expression given in Proposition \ref{prop:tgsvd} below.
	\begin{proposition}
		\label{prop:tgsvd}
		For GSVD-based precoding \cite{Chen2019}, (\ref{eqn:eppleq1}) simplifies as follows:
		\begin{align}
		\scalebox{0.9}{\mbox{\ensuremath{\displaystyle P_T = PL \left| \frac{1}{M_1+M_2-N}\right|.}}} \label{eqn:pgsvd}
		\end{align}
	\end{proposition}
	\begin{proof}
		Please refer to Appendix \ref{app:tgsvd}.
	\end{proof}
	
	\begin{remark}
		The expression given in Proposition \ref{prop:tgsvd} simplifies to \cite[Thm. 2]{Chen2019} for the asymptotic case $\frac{M_1}{N} = \frac{M_2}{N} = \eta, M_1,M_2,N \to \infty.$
	\end{remark} 

	The right hand side of (\ref{eqn:pgsvd}) increases unboundedly as $M_1+M_2 \to N,$ leading to an exceedingly large transmit power. Furthermore, when $M_1+M_2=N,$ the transmit power becomes infinite, rendering communication impossible. Moreover, from (\ref{eqn:cs}), the effective channels of $\bar{M}_1$ spatial streams of user 1 and $\bar{M}_2$ spatial streams of user 2 are forced to $\boldsymbol{I}_{\bar{M}_1}$ and $\boldsymbol{I}_{\bar{M}_2},$ respectively, cf. (\ref{eqn:gsvd}), (\ref{eqn:cs}), thereby necessitating channel inversion at the BS, which has a detrimental effect on performance. Hence, in the following, we develop the proposed UA-SD MIMO-NOMA scheme which overcomes the above limitations while still achieving SD of the users' channels.
	
	\section{The Proposed UA-SD MIMO-NOMA Scheme}
	\label{sec:proposed}
	In this section, we begin by describing a new matrix decomposition technique. Then, we utilize this matrix decomposition to develop the proposed UA-SD MIMO-NOMA precoding and decoding schemes. Lastly, we derive expressions for the achievable rates of the users.

	\subsection{Simultaneous Diagonalization}
	Let $\bar{M}_1, \bar{M}_2,$ and $M$ be as defined in Section \ref{sec:gsvd}. The proposed SD matrix decomposition is compactly stated in the following theorem.
		\begin{theorem}
		\label{thm:mgsvd}
		Let $\boldsymbol{H}_1$ and $\boldsymbol{H}_2$ be defined as in Section \ref{ssec:systemmodel}. Then, there exist unitary matrices $\boldsymbol{Q}_1 \in \mathbb{C}^{M_1\times M_1}$ and $\boldsymbol{Q}_2 \in \mathbb{C}^{M_2\times M_2},$ and a full matrix $\boldsymbol{Z} \in \mathbb{C}^{N\times L}$  such that
		\begin{align}
		\boldsymbol{Q}_1 \boldsymbol{H}_1 \boldsymbol{Z} = \begin{bmatrix*}[l]\boldsymbol{\Sigma}_1 & \boldsymbol{0} & \boldsymbol{0}\phantom{\hspace{4pt}}\\\boldsymbol{0} & \boldsymbol{D}_1 & \boldsymbol{0}\end{bmatrix*}, \qquad
		\boldsymbol{Q}_2 \boldsymbol{H}_2 \boldsymbol{Z} =  \begin{bmatrix*}[l] \boldsymbol{T} & \boldsymbol{0}\phantom{{}_1} & \boldsymbol{D}_2 \\	\boldsymbol{\Sigma}_2 & \boldsymbol{0} & \boldsymbol{0}	\end{bmatrix*}, \label{eqn:mgsvd}
		\end{align}
		where $\boldsymbol{\Sigma}_1, \boldsymbol{\Sigma}_2 \in \mathbb{R}^{M\times M}$ are diagonal matrices such that $\boldsymbol{\Sigma}_2 \boldsymbol{\Sigma}_1^{-1}$ contains the GSVs of $\boldsymbol{H}_2$ and $\boldsymbol{H}_1$ on its main diagonal\footnote{Note that the GSVs of $\boldsymbol{H}_2$ and $\boldsymbol{H}_1$ are the solutions $\mu \geq 0$ to the equation $\mathrm{det}\left(\boldsymbol{H}_2^\mathrm{H}\boldsymbol{H}_2 - \mu^2 \boldsymbol{H}_1^\mathrm{H} \boldsymbol{H}_1\right) = 0.$ They are the inverses of the GSVs of $\boldsymbol{H}_1$ and $\boldsymbol{H}_2,$ which are the solutions $\mu' \geq 0$ to $\mathrm{det}\left(\boldsymbol{H}_1^\mathrm{H}\boldsymbol{H}_1 - \mu'^2 \boldsymbol{H}_2^\mathrm{H} \boldsymbol{H}_2\right) = 0.$}, $\boldsymbol{T} \in \mathbb{C}^{\bar{M}_2\times M}$ is a full matrix, and $\boldsymbol{D}_k \in \mathbb{R}^{\bar{M}_k\times \bar{M}_k}, k=1,2,$ are diagonal matrices. $\boldsymbol{D}_1$ is defined as
		\begin{align}
			\boldsymbol{D}_1 = \begin{cases}
			\boldsymbol{S}_1 & \text{if $M_1+M_2 \leq N$} \\
			\boldsymbol{I}_{\bar{M}_1} & \text{otherwise,}
			\end{cases} \label{eqn:d1}
		\end{align}	
		where $\boldsymbol{S}_1$ contains the $\bar{M}_1$ singular values of $\frac{1}{\sqrt{\mathstrut\bar{M}_1}}\boldsymbol{H}_1\bar{\boldsymbol{H}}_2$ on its main diagonal. $\boldsymbol{D}_2$ contains the $\bar{M}_2$ singular values of $\frac{1}{\sqrt{\mathstrut\bar{M}_2}}\boldsymbol{H}_2\bar{\boldsymbol{H}}_1$ on its main diagonal. 
	\end{theorem}
	\begin{proof}
		Please refer to Appendix \ref{app:mgsvd}.
	\end{proof}
	
	In the following, we utilize the matrix decomposition in Theorem \ref{thm:mgsvd} to develop the proposed precoder and decoder designs.
	
	\subsection{Precoder Design}
	\label{sec:precoderdesign}	
	Matrices $\boldsymbol{Q}_1$ and $\boldsymbol{Q}_2$ from (\ref{eqn:mgsvd}) are used as detection matrices and the precoder matrix is chosen as $\boldsymbol{P} = \boldsymbol{Z}.$ Hence, based on Theorem \ref{thm:mgsvd}, the received signals at the users can be obtained from (\ref{eqn:yk}) as follows
	\begin{align}
	\boldsymbol{y}_1 &= \frac{1}{\sqrt{\Pi_1}}\begin{bmatrix*}[l]\boldsymbol{\Sigma}_1 & \boldsymbol{0} & \boldsymbol{0}\phantom{\hspace{4pt}}\\\boldsymbol{0} & \boldsymbol{D}_1 & \boldsymbol{0}\end{bmatrix*} \boldsymbol{s} + \boldsymbol{z}_1, \label{eqn:y1}\\
	\boldsymbol{y}_2 &= \frac{1}{\sqrt{\Pi_2}}\begin{bmatrix*}[l] \boldsymbol{T} & \boldsymbol{0}\phantom{{}_1} & \boldsymbol{D}_2 \\	\boldsymbol{\Sigma}_2 & \boldsymbol{0} & \boldsymbol{0}	\end{bmatrix*}  \boldsymbol{s} + \boldsymbol{z}_2, \label{eqn:y2i}
	\end{align}
	where the effective channel matrix of user 1 is diagonalized and the effective channel matrix of user 2 is partially diagonalized except for matrix $\boldsymbol{T}$ which causes self-interference for the first $\bar{M}_2$ elements of $\boldsymbol{y}_2.$ Note that the computation of $\boldsymbol{Q}_1, \boldsymbol{Q}_2, \boldsymbol{Z},$ and $\boldsymbol{T}$ is specified in Appendix \ref{app:mgsvd}.
	
	As seen from (\ref{eqn:y1}) and (\ref{eqn:y2i}), symbols $s_{l}, l=1,\dots,M,$ are received by both users. Symbols $s_{l}, l=M+1,\dots,M+\bar{M}_1,$ are received only by user 1 as the entries of the corresponding columns of the effective channel matrix of user 2 in (\ref{eqn:y2i}) are zero, hence, we set $p_{2,l} = 0$ for $l=M+1,\dots,M+\bar{M}_1.$ Furthermore, symbols $s_{l}, l=M+\bar{M}_1+1,\dots,L,$ are received only by user 2 as the entries of the corresponding columns of the effective channel matrix of user 1 in (\ref{eqn:y1}) are zero, hence, we set $p_{1,l} = 0$ for $l=M+\bar{M}_1+1,\dots,L.$

	\subsection{Decoding Scheme}
	\label{sec:decoding}
	At user 1, from (\ref{eqn:y1}), as the channel is diagonalized and $p_{2,l} = 0$ for $l=M+1,\dots,M+\bar{M}_1,$ symbols $s_{1,l}, l=M+1,\dots,M+\bar{M}_1,$ are decoded directly from the last $\bar{M}_1$ elements of $\boldsymbol{y}_1.$ Next, if $M > 0,$ symbols $s_{1,l}, l=1,\dots,M,$ are decoded from the first $M$ elements of $\boldsymbol{y}_1$ treating the symbols of the second user, $s_{2,l}, l=1,\dots, M,$ as interference as in SISO-NOMA \cite{Saito2013}.
	
	At user 2, if $M > 0,$ from (\ref{eqn:y2i}), using SIC both $s_{1,l}$ and $s_{2,l}, l=1,\dots,M,$ are decoded from the last $M$ elements of $\boldsymbol{y}_2$ as in SISO-NOMA \cite{Saito2013}. Next, self-interference is canceled by subtracting the reconstructed interference, $\boldsymbol{T} \hat{\boldsymbol{s}},$ from the first $\bar{M}_2$ elements of $\boldsymbol{y}_2,$ resulting in the self-interference-free signal
	\begin{align}
	\tilde{\boldsymbol{y}}_2 &= \boldsymbol{y}_2 - \frac{1}{\sqrt{\Pi_2}} \begin{bmatrix}
	\boldsymbol{T} \hat{\boldsymbol{s}}  \\
	\boldsymbol{0} 
	\end{bmatrix}  = \frac{1}{\sqrt{\Pi_2}} \begin{bmatrix}
		\boldsymbol{0} & \boldsymbol{0} & \boldsymbol{D}_2 \\
		\boldsymbol{\Sigma}_2 & \boldsymbol{0} & \boldsymbol{0}
		\end{bmatrix}  \boldsymbol{s} + \boldsymbol{z}_2, \label{eqn:y2}
	\end{align}
	where the elements of $\hat{\boldsymbol{s}} \in \mathbb{C}^{M\times 1}$ are the previously decoded symbols\footnote{We assume that the symbols $s_{k,l},k=1,2,l=1,\dots,M,$ have been correctly decoded. This assumption is justified so long as the rates corresponding to the symbols are at or below the achievable rate $R_{k,l}$ provided in Section \ref{sec:ar}.} $s_l = p_{1,l} s_{1,l} + p_{2,l} s_{2,l}, l=1,\dots,M.$ Lastly, symbols $s_{2,l}, l=M+\bar{M}_1+1,\dots,L,$ are decoded from the first $\bar{M}_2$ elements of $\tilde{\boldsymbol{y}}_2,$ which are interference free as $p_{1,l} = 0$ for $l=M+\bar{M}_1+1,\dots,L.$

	\begin{remark}
		When $\bar{M}_2 > 0,$ in order to eliminate the self-interference due to $\boldsymbol{T},$ both $s_{1,l}$ and $s_{2,l}, l=1,\dots,M,$ must be decoded at user 2. Therefore, SIC has to be performed at user 2 for all SISO-NOMA symbols $s_l,$ $l=1,\dots,M.$ This is different from GSVD-based precoding in \cite{Chen2019}, where SIC is performed at user 1 or 2 depending on the values of $[\boldsymbol{C}_1]_{ll}$ and $[\boldsymbol{S}_1]_{ll},l=1,\dots,M,$ given in (\ref{eqn:cs}). On the other hand, when $\bar{M}_2 = 0,$ as there is no self-interference due to $\boldsymbol{T},$ SIC can potentially be performed at user 1 or 2 depending on the values of $[\boldsymbol{\Sigma}_1]_{ll}$ and $[\boldsymbol{\Sigma}_2]_{ll}, l=1,\dots,M,$ similar to GSVD-based MIMO-NOMA. Nevertheless, we show in Section \ref{sec:epa} that, for the proposed scheme, the performance gain obtained by allowing flexible SIC at user 1 or 2 is insignificant when the users are located sufficiently far apart, which is the most relevant scenario for MIMO-NOMA.
	\end{remark}
	
	Next, we calculate the computational complexity of the proposed UA-SD MIMO-NOMA and compare it with that of GSVD-based MIMO-NOMA \cite{Chen2019}.
	
	\subsection{Computational Complexity}	
	For the proposed scheme, the required operations for the matrix decomposition in Theorem \ref{thm:mgsvd}, based on Appendix \ref{app:mgsvd}, are summarized in Table \ref{tab:cc}. As the QR decomposition and SVD of an $M\times N$ matrix entail complexities of $\mathcal{O}\mkern-\medmuskip\left(2N^2(M-\frac{N}{3})+4(M^2 N - M N^2 + \frac{N^3}{3})\right)$ and $\mathcal{O}\mkern-\medmuskip\left(4 M^2 N + 8 M N^2 + 9 N^3\right),$ respectively \cite[Sec. 5.5]{Bai1992}, obtaining $\boldsymbol{Q}_1, \boldsymbol{Q}_2,$ and $\boldsymbol{Z}$ based on Theorem \ref{thm:mgsvd} entails a total complexity of $\mathcal{O}\mkern-\medmuskip\left(\dfrac{134}{3}N^3\right)$ for $M_1=M_2=N.$ Furthermore, in the proposed scheme, self-interference cancellation at user 2 entails an additional complexity of $\mathcal{O}\mkern-\medmuskip\left(M\bar{M}_2\right).$ On the other hand, GSVD used in \cite{Chen2019} requires a QR decomposition of size $(M_1+M_2)\times N$ and a cosine-sine (CS) decomposition \cite{VanLoan1985} of size $M_1+M_2.$ The CS decomposition entails a complexity of $\mathcal{O}\mkern-\medmuskip\left(36 N^3\right),$ resulting in an overall complexity of $\mathcal{O}\mkern-\medmuskip\left(\dfrac{116}{3} N^3\right)$ for $M_1=M_2=N.$ Hence, both the proposed UA-SD MIMO-NOMA and GSVD-based MIMO-NOMA \cite{Chen2019} have an overall complexity order of $\mathcal{O}\mkern-\medmuskip\left(N^3\right).$

	\begin{table}
		\centering
		\caption{Required operations for decomposition in Theorem \ref{thm:mgsvd}, see Appendix \ref{app:mgsvd} for definition of matrices.}
		\label{tab:cc}
		\begin{tabular}{|l|l|c|l||l|l|c|l|}
			\hline
			& Computation & Dimension & Operation & & Computation & Dimension & Operation \\\hline
			1) & $\bar{\boldsymbol{H}}_1$ & $N\times M_1$ & QR decomposition & 4) & $\hat{\boldsymbol{U}}_1\boldsymbol{\Sigma}_1\hat{\boldsymbol{V}}_1^\mathrm{H}$ & $M_1\times N$ & SVD \\\hline
			2) & $\bar{\boldsymbol{H}}_2$ & $N\times M_2$ & QR decomposition & 5) & $\tilde{\boldsymbol{H}}_2^\mathrm{H}  = \boldsymbol{Q}\boldsymbol{R}$ & $M_1\times M_2$ & QR decomposition \\\hline
			3) & $\boldsymbol{K}$ & $N\times (\bar{M}_1+\bar{M}_2)$ & QR decomposition & 6) & $\boldsymbol{\Sigma}$ & $M\times M$ & SVD \\\hline
		\end{tabular}
		\vspace{-0.5cm}
	\end{table}

	\subsection{Achievable Rates}
	\label{sec:ar}
	Based on (\ref{eqn:y1}), the achievable rates for the $s_{1,l}, l=1,\dots,L,$ of user 1 are given by
	\begin{align}
	R_{1,l}^{(1)} = \begin{cases}
	\log_2\left(1+\frac{p_{1,l}}{\Pi_1}\frac{([\boldsymbol{\Sigma}_1]_{ll})^2}{\sigma^2+\frac{p_{2,l}}{\Pi_1}([\boldsymbol{\Sigma}_1]_{ll})^2}\right)	& \text{for $l=1,\dots,M$} \\
	\log_2\left(1+\frac{p_{1,l}}{\Pi_1} \frac{([\boldsymbol{D}_1]_{ll})^2}{\sigma^2}\right) & \text{for $l=M+1,\dots,M+\bar{M}_1$} \\
	0 & \text{otherwise.}\label{eqn:r11n}
	\end{cases}
	\end{align}
	
	At user 2, in order to perform SIC, the $s_{1,l},l=1,\dots,M,$ are decoded first. From (\ref{eqn:y2i}), the achievable rates for the $s_{1,l}, l=1\dots,M,$ at user 2 are given by	
	\begin{align}
	\scalebox{0.9}{\mbox{\ensuremath{\displaystyle R_{1,l}^{(2)} = \log_2\left(1+\frac{p_{1,l}}{\Pi_2} \frac{([\boldsymbol{\Sigma}_2]_{ll})^2}{\sigma^2+\frac{p_{2,l}}{\Pi_2}([\boldsymbol{\Sigma}_2]_{ll})^2}\right)\label{eqn:r12nf}.}}}
	\end{align}
	
	For the decoding of the $s_{1,l}, l=1,\dots,M,$ to be successful at both user 1 and 2 \emph{for every channel use}, the code rate is chosen as follows:
	\begin{equation}
	R_{1,l} = \mathrm{min}\left\{R_{1,l}^{(1)},R_{1,l}^{(2)}\right\} = \begin{cases}
		R_{1,l}^{(1)} & \text{if $\frac{([\boldsymbol{\Sigma}_1]_{ll})^2}{\Pi_1} <  \frac{([\boldsymbol{\Sigma}_2]_{ll})^2}{\Pi_2}$} \\
		R_{1,l}^{(2)} & \text{otherwise.}
	\end{cases} \label{eqn:minr11r12}
	\end{equation}
		
	Next, from (\ref{eqn:y2i}) and (\ref{eqn:y2}), the achievable rates for the $s_{2,l}, l=1,\dots,L,$ after SIC and self-interference cancellation are given by
	\begin{align}
	R_{2,l} = \begin{cases} \log_2\left(1+\frac{p_{2,l}}{\Pi_2}\frac{([\boldsymbol{\Sigma}_2]_{ll})^2}{\sigma^2}\right)	&\text{for $l=1,\dots,M$} \\[5pt]
	\log_2\left(1+\frac{p_{2,l}}{\Pi_2}\frac{([\boldsymbol{D}_2]_{ll})^2}{\sigma^2}\right) &\text{for $l=M+\bar{M}_1+1,\dots,L$} \\
	0 & \text{otherwise.}\label{eqn:r2nf}
	\end{cases}
	\end{align}
	Lastly, the ergodic achievable sum rates for users 1 and 2, $R_1$ and $R_2,$ respectively, are given by
	\begin{align}
	\scalebox{1}{\mbox{\ensuremath{\displaystyle R_1}}} \scalebox{1}{\mbox{\ensuremath{\displaystyle = \mathrm{E}\left[\sum_{l=1}^{M} R_{1,l}\right] + \mathrm{E}\left[\sum_{l=M+1}^{M+\bar{M}_1} R_{1,l}^{(1)}\right],}}} \quad
	\scalebox{1}{\mbox{\ensuremath{\displaystyle R_2}}} \scalebox{1}{\mbox{\ensuremath{\displaystyle = \mathrm{E}\left[\sum_{l=1}^{M} R_{2,l}\right] + \mathrm{E}\left[\sum_{l=M+\bar{M}_1+1}^{L} R_{2,l}\right].}}} \label{eqn:r}
	\end{align}
	
	\begin{remark}
		The diagonal entries of $\boldsymbol{\Sigma}_1$ and $\boldsymbol{\Sigma}_2$ contain the GSVs of $\boldsymbol{H}_2$ and $\boldsymbol{H}_1,$ c.f. Theorem \ref{thm:mgsvd}. In \cite{Chen2019}, analogous ergodic rate expressions for GSVD-based NOMA with EPA are simplified using the asymptotic probability density function (pdf) of GSVs. However, simplification of (\ref{eqn:r}) for finite $M_1, M_2,$ and $N$ necessitates a novel approach based on the finite-size marginal and ordered pdfs of the GSVs, which are derived in Section \ref{sec:performance}.
	\end{remark}

	Next, we describe a hybrid scheme which further improves the ergodic achievable rate of UA-SD MIMO-NOMA for user rate pairs close to the single user rates, see Section \ref{sec:sim}

	\subsection{Hybrid UA-SD MIMO-NOMA}
	\label{sec:hybrid}
	
		Due to sustained channel inversion in the proposed UA-SD MIMO-NOMA scheme, see Appendix \ref{app:mgsvd}, in general, the single-user (SU)-MIMO rates of the users cannot be achieved. This results in an inferior performance compared to TDMA-based MIMO-OMA for user rate pairs close to the single-user rates. In order to mitigate this effect, we propose a hybrid scheme employing time sharing between SU-MIMO and UA-SD MIMO-NOMA. In the hybrid scheme, due to time sharing, all user rate pairs which are convex combinations of UA-SD MIMO-NOMA and TDMA-based MIMO-OMA rate pairs can be achieved, thereby enhancing the region of achievable user rates, see Section \ref{sec:sim}.

	\section{Performance Analysis and Power Allocation}
	\label{sec:performance}
	In this section, the ergodic achievable rate expressions of the proposed UA-SD MIMO-NOMA, provided in (\ref{eqn:r}), are simplified via finite-size RMT for EPA and UPA. Furthermore, a power allocation algorithm for UPA is presented.	
	\subsection{Equal Power Allocation}
	\label{sec:epa}
	For EPA, the same power $P$ is allocated to all symbols $s_l, l=1,\dots,L.$ Furthermore, for $l=1,\dots,M,$ $P$ is partitioned between users 1 and 2 as $p_{1,l}=P_1$ and $p_{2,l}=P_2,$ respectively, such that $P_1+P_2=P.$ For EPA, the transmit power in (\ref{eqn:eppleq1}) can be simplified as follows.

	\begin{proposition}
		\label{prop:pwae}
		For EPA, (\ref{eqn:eppleq1}) simplifies as follows:
		\begin{align}
			P_T = \begin{cases}
			2P & \text{if $M_1+M_2 \leq N$} \\
			\frac{PN}{M_1+M_2-N}	& \text{if $M_1 + M_2 > N$ and $M_1 \geq N$} \\
			P\left(\frac{M_1}{N} + \frac{\bar{M}_1}{M} + 1\right) & \text{otherwise.} 
			\end{cases} \label{eqn:fpap1}
		\end{align}
	\end{proposition}
	\begin{proof}
	For $M_1+M_2 \leq N,$ based on (\ref{eqn:bdpt}), we have $P_T = 2P.$ Next, for $M_1+M_2 > N,$ (\ref{eqn:eppleq1}) reduces to $P_T = P\mathrm{E}\left[\mathrm{tr}\left(\boldsymbol{Z}^H \boldsymbol{Z}\right)\right],$ which simplifies to (\ref{eqn:fpap1}) based on \cite[Thm. 3.3.8]{Nagar2011}.
	\end{proof}

	Next, we focus on simplifying the ergodic rate expressions of the users for EPA. From Section \ref{sec:ar}, we observe that the ergodic rate expressions depend on $[\boldsymbol{\Sigma}_k]_{ll}, k=1,2, l=1,\dots,M,$ and  $[\boldsymbol{D}_k]_{ll}, k=1,2,l=1,\dots,\bar{M}_k.$ However, $[\boldsymbol{\Sigma}_k]_{ll}$ and $[\boldsymbol{D}_k]_{ll}$ are not independent for different $l,$ respectively. Therefore, the evaluation of the ergodic rate expressions involving $[\boldsymbol{\Sigma}_k]_{ll}$ and $[\boldsymbol{D}_k]_{ll}$ require $M$- and $\bar{M}_k$-dimensional integration, respectively, which introduces a very high computational complexity. Fortunately, the computational complexity can be drastically reduced by utilizing the marginal eigenvalue pdfs introduced in the following definition.

	
	\begin{definition}[Based on \cite{Zanella2009}]
		\label{def:mepdf}
		Let $\boldsymbol{X}$ be a $q\times q$ Hermitian symmetric random matrix with eigenvalues $\lambda_1 \geq \lambda_2 \geq \dots \geq \lambda_q,$ and let $\boldsymbol{\lambda} = \{\lambda_n,n=1,\dots,q\}.$ Furthermore, let the unordered eigenvalue of $\boldsymbol{X}$ be denoted by $\lambda(\boldsymbol{X})$ and the joint pdf of the eigenvalues by $p_{\boldsymbol{\lambda}}(\boldsymbol{\lambda}).$ Then, the marginal eigenvalue pdf is defined as
		\begin{align}
		p_{\lambda(\boldsymbol{X})}(\lambda) = \left. \int_{0}^{+\infty} \dots \int_{0}^{+\infty} p_{\boldsymbol{\lambda}}(\boldsymbol{\lambda}) \mathrm{d}\lambda_q \mathrm{d}\lambda_{q-1} \dots \mathrm{d}\lambda_2 \, \biggm|_{\lambda_1 \to \lambda} \right..
		\end{align}
	\end{definition}
	
	Using the above definition, for a $q\times q$ Hermitian symmetric random matrix $\boldsymbol{X}$ with eigenvalues $\lambda_1, \lambda_2, \dots, \lambda_q,$ we may simplify the expectation of a generic additively separable function
	\begin{align}
		f(\lambda_1,\dots,\lambda_q) = \sum_{n = 1}^{q} g(\lambda_n) \label{eqn:addsep}
	\end{align}
	in the eigenvalues of $\boldsymbol{X}$ as
	\begin{align}
	\scalebox{0.85}{\mbox{\ensuremath{\displaystyle \mathrm{E}\left[f(\lambda_1,\dots,\lambda_q)\right]}}} &= \scalebox{0.85}{\mbox{\ensuremath{\displaystyle \int_{0}^{+\infty} \dots \int_{0}^{+\infty} f(\lambda_1,\dots,\lambda_q) p_{\boldsymbol{\lambda}}(\boldsymbol{\lambda}) \mathrm{d}\lambda_q \dots\mathrm{d}\lambda_1 =  \sum_{n=1}^{q} \int_{0}^{+\infty} \dots \int_{0}^{+\infty} g(\lambda_n) p_{\boldsymbol{\lambda}}(\boldsymbol{\lambda}) \mathrm{d}\lambda_q \dots\mathrm{d}\lambda_1}}} \nonumber\\
	&\overset{(a)}{=} \scalebox{0.85}{\mbox{\ensuremath{\displaystyle q \int_{0}^{+\infty} g(\lambda) p_{\lambda(\boldsymbol{X})}(\lambda) \mathrm{d}\lambda,}}} \label{eqn:mpdfe}
	\end{align}
	where (a) is obtained by using Definition \ref{def:mepdf} and exploiting the symmetry in the integration variable. Before we can exploit (\ref{eqn:mpdfe}) to calculate the ergodic rates of the users, we need to characterize the diagonal elements of $\boldsymbol{\Sigma}_1$ and $\boldsymbol{\Sigma}_2,$ cf. Section \ref{sec:ar}.

	\begin{proposition}
	\label{prop:equif}
		When $M_1+M_2 > N,$ the squares of the $l$-th diagonal elements, $l=1,\dots,M,$ of matrices $\boldsymbol{\Sigma}_1$ and $\boldsymbol{\Sigma}_2$ are given by
		\begin{align}
			[\boldsymbol{\Sigma}_1]_{ll}^2 = \frac{1}{\lambda_l+1},\qquad [\boldsymbol{\Sigma}_2]_{ll}^2 = \frac{\lambda_l}{\lambda_l+1},\label{eqn:equif}
		\end{align}
		where $\lambda_l$ is the $l$-th ordered eigenvalue of the F-distributed matrix $\boldsymbol{F} = (\boldsymbol{W}_2)^{\frac{1}{2}} \boldsymbol{W}_1^{-1} (\boldsymbol{W}_2)^{\frac{1}{2}}$ \cite{Perlman1977, Forrester2014}, and $\boldsymbol{W}_2 \sim \mathcal{CW}_{\nu}(\mu_2,\boldsymbol{I}_{\nu})$ and $\boldsymbol{W}_1 \sim \mathcal{CW}_{\nu}(\mu_1, \boldsymbol{I}_{\nu})$ are independent Wishart distributed matrices. Coefficients $\mu_1, \mu_2,$ and $\nu$ are given in Table \ref{tab:map}.
		
		\begin{table}
		\centering
		\caption{Wishart matrix parameters $\mu_1,\mu_2,\nu$ for different $M_1,M_2,$ and $N.$}
		\label{tab:map}
		\begin{tabular}{|l|l||l|l|}
			\hline
			Condition & $(\mu_1,\mu_2,\nu)$ & Condition & $(\mu_1,\mu_2,\nu)$ \\\hline\hline
			$M_1,M_2 \geq N$ & $(M_1,M_2,N)$ & 	$M_1,M_2 < N, M_1+M_2 > N$ & $(M_1,M_2,M)$ \\\hline
			$M_1 \geq N, M_2 < N$ & $(M_1+M_2-N,N,M_2)$ & 	$M_1 < N, M_2 \geq N$ & $(N,M_1+M_2-N,M_1)$ \\\hline
		\end{tabular}
		\vspace{-0.5cm}
		\end{table}
	\end{proposition}
	\begin{proof}
		Please refer to Appendix \ref{app:equif}.
	\end{proof}
	
	\begin{remark}
		When $M_1+M_2 \leq N,$ as $M=0,$ matrices $\boldsymbol{\Sigma}_1$ and $\boldsymbol{\Sigma}_2$ are empty.
	\end{remark}
	
	Next, in order to obtain the marginal pdfs of $[\boldsymbol{\Sigma}_k]_{ll}, k=1,2, l=1,\dots,M,$ we utilize the marginal eigenvalue pdf of an $\boldsymbol{F}$ distributed matrix, $p_{\lambda(\mathfrak{F})}(\lambda; \mu_1, \mu_2, \nu),$ in Theorem \ref{thm:f}, which is then used to simplify the ergodic rates $R_{1,l}$ and $R_{2,l}, l=1,\dots,M,$ for EPA, via Definition \ref{def:mepdf}.
	
	\begin{theorem}
		\label{thm:f}
		Let $\boldsymbol{X} \sim \mathcal{CW}_q(m_1,\boldsymbol{I}_q)$ and $\boldsymbol{Y} \sim \mathcal{CW}_q(m_2,\boldsymbol{I}_q),$ $m_1,m_2 \geq q,$ be $q\times q$ independent complex-valued Wishart random matrices. Then, the marginal eigenvalue pdf of matrix $\boldsymbol{F} = \boldsymbol{Y}^{\frac{1}{2}}\boldsymbol{X}^{-1}\boldsymbol{Y}^{\frac{1}{2}}$ is given by
		\begin{align}
		&p_{\lambda(\mathfrak{F})}(\lambda; m_1, m_2, q) = K_\mathfrak{F} \frac{1}{(1+\lambda)^{m_1+m_2}} \sum_{m = 1}^{q}\sum_{n = 1}^{q} (-1)^{(n+m)}\lambda^{m+n-2+m_2-q} \mathrm{det}\left(\boldsymbol{\Xi}^{[m,n]}\right), \label{eqn:f}
		\end{align}
		where $K_\mathfrak{F}$ is a constant ensuring that the integral over the pdf is equal to one, and the elements of the $(q-1)\times (q-1)$ matrix $\boldsymbol{\Xi}^{[m,n]}$ are given by
		\begin{equation}
		\left[\boldsymbol{\Xi}^{[m,n]}\right]_{ij} = \mathrm{B}(m_2-q+\alpha(i,j,m,n)+1, m_1+q-\alpha(i,j,m,n)-1),
		\end{equation}
		for $i,j=1,\dots,q-1,$ where $\mathrm{B}(\cdot,\cdot)$ denotes the Beta function, and
		\begin{align}
		\alpha(i,j,m,n) &= \begin{cases}
			i+j-2	& \text{if $i<m$ and $j<n$} \\
			i+j	& \text{if $i \geq m$ and $j \geq n$} \\
			i+j-1 & \text{otherwise.}
		\end{cases}
		\end{align}
		The support of the pdf is $\lambda \in [0,+\infty).$
	\end{theorem}
	\begin{proof}
		Please refer Appendix \ref{app:f}.
	\end{proof}	

	\begin{remark} Theorem \ref{thm:f} is the finite-size counterpart to \cite[Thm. 1]{Chen2019}. Furthermore, the finite-size marginal eigenvalue pdf for the special case $M_1 = M_2$ has been provided in \cite[Thm. 2]{Chen2019a}.\end{remark}

	Now, in order to obtain the marginal pdfs of $[\boldsymbol{D}_k]_{ll}, k=1,2,l=1,\dots,\bar{M}_k,$ which are used to simplify $R_{1,l}^{(1)}, l = M+1,\dots,M_1,$ and $R_{2,l}, l=M_1+1,\dots,L,$ using Definition \ref{def:mepdf}, we provide the following results.

	\begin{proposition}
		\label{prop:d1d2}
		The squares of the $l$-th diagonal elements of $\boldsymbol{D}_1$ and $\boldsymbol{D}_2$ are given by
		\begin{align}
			([\boldsymbol{D}_1]_{ll})^2 = \begin{cases}
				1	& \text{if $M_1+M_2 > N$} \\
				\lambda_l^{(1)} & \text{otherwise,}
			\end{cases}, \qquad 	([\boldsymbol{D}_2]_{ll})^2 = \lambda_l^{(2)},
		\end{align}
		where $\lambda_l^{(1)}$ and $\lambda_l^{(2)}$ follow the same distributions as the $l$-th eigenvalues of the Wishart matrices $\boldsymbol{W}_1 \sim \mathcal{CW}_{\bar{M}_1}(M_1,\frac{1}{\bar{M}_1}\boldsymbol{I}_{\bar{M}_1})$ and $\boldsymbol{W}_2 \sim \mathcal{CW}_{\bar{M}_2}(M_2,\frac{1}{\bar{M}_2}\boldsymbol{I}_{\bar{M}_2}),$ respectively.
	\end{proposition}
	\begin{proof}
	If $M_1 + M_2 \leq N,$ as the $([\boldsymbol{D}_1]_{ll})^2$ are the squared singular values of $\frac{1}{\sqrt{\mathstrut\bar{M}_1}}\boldsymbol{H}_1\bar{\boldsymbol{H}}_2,$ cf. Theorem \ref{thm:mgsvd}, their distribution is identical to the distribution of the eigenvalues of Wishart matrix $\boldsymbol{W}_1 \sim \mathcal{CW}_{\bar{M}_1}(M_1,\frac{1}{\bar{M}_1}\boldsymbol{I}_{\bar{M}_1})$ \cite[Thm. 3.2.4]{Gupta1999}. Otherwise, if $M_1 + M_2 > N,$ from (\ref{eqn:d1}), $([\boldsymbol{D}_1]_{ll})^2 = 1.$ Similarly, the distribution of $([\boldsymbol{D}_2]_{ll})^2$ is identical to the distribution of the eigenvalues of Wishart matrix $\boldsymbol{W}_2 \sim \mathcal{CW}_{\bar{M}_2}(M_2,\frac{1}{\bar{M}_2}\boldsymbol{I}_{\bar{M}_2}).$
	\end{proof}

	The marginal eigenvalues of $\boldsymbol{W}_1$ and $\boldsymbol{W}_2,$ denoted by $p_{\lambda(\mathfrak{W})}(\lambda; M_1,\bar{M}_1)$ and $p_{\lambda(\mathfrak{W})}(\lambda; M_2,\bar{M}_2),$ respectively, can be obtained from Theorem \ref{thm:e}.
	
	\begin{theorem}
		\label{thm:e}
		Let $\boldsymbol{W}$ be a $q\times q$ complex-valued Wishart matrix $\boldsymbol{W} \sim \mathcal{CW}_q(p,\frac{1}{q}\boldsymbol{I}_q),$ $p \geq q,$ then the marginal eigenvalue pdf of $\boldsymbol{W}$ is given by
		\begin{align}
		p_{\lambda(\mathfrak{W})}(\lambda; p, q) &= K_\mathfrak{W} \sum_{m = 1}^{q}\sum_{n = 1}^{q} (-1)^{(n+m)} (q\lambda)^{n+m-2+p-q} \mathrm{exp}\left(-q\lambda\right) \mathrm{det}\left(\boldsymbol{\Omega}^{[m,n]}\right), \label{eqn:e}
		\end{align}
		where $K_\mathfrak{W}$ is a constant ensuring that the integral over the pdf is equal to one, and the elements of $(q-1)\times (q-1)$ matrix $\boldsymbol{\Omega}^{[m,n]}$ are given by
		\begin{equation}
		\left[\boldsymbol{\Omega}^{[m,n]}\right]_{ij} = \left(\alpha(i,j,m,n) + p-q\right)!,
		\end{equation}
		for $i,j = 1,\dots,q-1,$ and $\alpha(i,j,m,n)$ is defined as in Theorem \ref{thm:f}. The support of the pdf is $\lambda \in [0,+\infty).$
	\end{theorem}
	\begin{proof}
		Eq. (\ref{eqn:e}) follows directly from \cite[Sec. IV.A]{Zanella2009}.
	\end{proof}
	
	Now, we are ready to present the ergodic user rate expressions for EPA in the following proposition.
	
	\begin{proposition}
		\label{prop:eq}
		For EPA and $M_1+M_2\leq N,$ the ergodic rates of users 1 and 2 for the proposed UA-SD MIMO-NOMA are given by
		\begin{align}
		\scalebox{0.9}{\mbox{\ensuremath{\displaystyle R_1}}} &= \scalebox{0.9}{\mbox{\ensuremath{\displaystyle \bar{M}_1\int_{0}^{+\infty} \log_2\left(1+\frac{P\lambda}{\Pi_1\sigma^2}\right) p_{\lambda(\mathfrak{W})}(\lambda; {M_1,\bar{M}_1}) \mathrm{d}\lambda}}}, \label{eqn:r1n1}\\
		\scalebox{0.9}{\mbox{\ensuremath{\displaystyle R_2}}} &= \scalebox{0.9}{\mbox{\ensuremath{\displaystyle \bar{M}_2\int_{0}^{+\infty} \log_2\left(1+\frac{P\lambda}{\Pi_2\sigma^2}\right) p_{\lambda(\mathfrak{W})}(\lambda; {M_2,\bar{M}_2}) \mathrm{d}\lambda}}}, \label{eqn:r2n1}
		\end{align}
		and for EPA and $M_1+M_2 > N,$ the rates are given by
		\begin{align}
		\scalebox{0.9}{\mbox{\ensuremath{\displaystyle R_1}}} &= \scalebox{0.9}{\mbox{\ensuremath{\displaystyle M\int_{0}^{\Pi^{-1}} \log_2\left(1+\frac{1}{\Pi_2}\frac{\frac{\lambda P_{1}}{1+\lambda}}{\sigma^2+\frac{1}{\Pi_2}\frac{\lambda P_{2}}{1+\lambda}}\right) p_{\lambda(\mathfrak{F})}(\lambda; {\mu_1,\mu_2,\nu}) \mathrm{d}\lambda}}} \nonumber\\
		&\qquad + \scalebox{0.9}{\mbox{\ensuremath{\displaystyle M \int_{\Pi^{-1}}^{+\infty} \log_2\left(1+\frac{1}{\Pi_1}\frac{\frac{P_{1}}{1+\lambda}}{\sigma^2+\frac{1}{\Pi_1}\frac{P_{2}}{1+\lambda}}\right) p_{\lambda(\mathfrak{F})}(\lambda; {\mu_1,\mu_2,\nu}) \mathrm{d}\lambda + \bar{M}_1 \log_2\left(1+\frac{P}{\Pi_1\sigma^2}\right)}}}, \label{eqn:r1n}\\
		\scalebox{0.9}{\mbox{\ensuremath{\displaystyle R_2}}} &= \scalebox{0.9}{\mbox{\ensuremath{\displaystyle M\int_{0}^{+\infty} \log_2\left(1+\frac{1}{\Pi_2}\frac{\lambda P_{2}}{(1+\lambda)\sigma^2}\right) p_{\lambda(\mathfrak{F})}(\lambda; {\mu_1,\mu_2,\nu}) \mathrm{d}\lambda}}} \nonumber\\
		&\hspace{3cm}\scalebox{0.9}{\mbox{\ensuremath{\displaystyle + \bar{M}_2 \int_{0}^{+\infty} \log_2\left(1+\frac{P\lambda}{\Pi_2\sigma^2}\right) p_{\lambda(\mathfrak{W})}(\lambda; {M_2,\bar{M}_2}) \mathrm{d}\lambda}}}, \label{eqn:r2n}
		\end{align}
		where $\mu_1,\mu_2,$ and $\nu$ are given in Table \ref{tab:map}, and $\Pi= \frac{\Pi_1}{\Pi_2}.$
	\end{proposition}
	\begin{proof}
	For $M_1+M_2 \leq N,$ (\ref{eqn:r}) simplifies to (\ref{eqn:r1n1}) and (\ref{eqn:r2n1}) based on Proposition \ref{prop:d1d2} and Definition \ref{def:mepdf}. For $M_1 + M_2 > N,$ (\ref{eqn:minr11r12}) reduces to
	\begin{equation}
	R_{1,l} = \begin{cases}	R_{1,l}^{(1)} & \text{if $\lambda_l > \Pi^{-1}$} \\
	R_{1,l}^{(2)} & \text{otherwise,} \label{eqn:r12min}
	\end{cases}
	\end{equation} 
	using which (\ref{eqn:r}) simplifies to (\ref{eqn:r1n}) and (\ref{eqn:r2n}) based on Propositions \ref{prop:equif} and \ref{prop:d1d2} and 
	Definition \ref{def:mepdf}.
	\end{proof}	

	\textbf{Impact of allowing SIC only at user 2:}
	In the proposed UA-SD MIMO-NOMA, unlike GSVD-based MIMO-NOMA in \cite{Chen2019}, only user 2 performs SIC, see Section \ref{sec:decoding}. Hence, when $\lambda < \Pi^{-1},$ the achievable rate of user 1, $R_{1,l},l=1,\dots,M,$ is limited by the inferior achievable rate at user 2, $R_{1,l}^{(2)},l=1,\dots,M,$ cf. (\ref{eqn:r12min}). Based on  (\ref{eqn:r12min}), the potential gain of the achievable ergodic rate of user 1 obtained by allowing SIC at user 1 or 2, as in GSVD-based MIMO-NOMA \cite{Chen2019}, can be upper bounded as follows:
	\begin{align}
	\scalebox{0.85}{\mbox{\ensuremath{\displaystyle U_1 = \mathrm{E}\left[\sum_{l=1}^{M} \left(R_{1,l}^{(1)}-R_{1,l}\right)\right]}}} &\leq \scalebox{0.8}{\mbox{\ensuremath{\displaystyle \mathrm{E}_{0 \leq \lambda \leq \Pi^{-1}}\left[\sum_{l=1}^{M} R_{1,l}^{(1)}\right] \overset{(a)}{=} M\int_{0}^{\Pi^{-1}} \underbrace{\log_2\left(1+\frac{1}{\Pi_1}\frac{\frac{P_{1}}{1+\lambda}}{\sigma^2+\frac{1}{\Pi_1}\frac{P_{2}}{1+\lambda}}\right)}_{\coloneqq R_\mathrm{U}(\lambda)} p_{\lambda(\mathfrak{F})}(\lambda; {\mu_1,\mu_2,\nu}) \mathrm{d}\lambda}}} \nonumber\\
	&\overset{(b)}{\leq} \scalebox{0.85}{\mbox{\ensuremath{\displaystyle M  K_\mathrm{U} \underbrace{\int_{0}^{\Pi^{-1}} p_{\lambda(\mathfrak{F})}(\lambda; {\mu_1,\mu_2,\nu})\mathrm{d}\lambda}_{\mathrm{Pr}\left\{\lambda < \Pi^{-1}\right\}}}}}, \label{eqn:u1}
	\end{align}
	where $K_\mathrm{U} = \max_{0 \leq \lambda \leq \Pi^{-1}}\left\{R_\mathrm{U}(\lambda)\right\} < \infty$ is a bounded constant independent of $\lambda,$ and (a) and (b) follow from Definition \ref{def:mepdf} and the H\"{o}lder's inequality, respectively. From (\ref{eqn:u1}), we observe that the integral tends to zero as $\Pi \to \infty.$ Hence, the potential performance gain also approaches zero as $\Pi \to \infty.$

	\textbf{Approximations based on asymptotic pdfs:}
	A large finite-dimensional approximation of the marginal eigenvalue pdfs for $\boldsymbol{F}$ and $\boldsymbol{W},$ based on the asymptotic analysis in \cite{Chen2019}, can be given as follows. Let $\boldsymbol{F}$ be as defined in Theorem \ref{thm:f}. Furthermore, let $\rho_1 = \frac{q}{m_1}$ and $\rho_2 = \frac{q}{m_2}, m_1,m_2 > q.$ Then, based on \cite[Thm. 2.30]{Tulino2004},\cite{Chen2019}, an approximation of the marginal eigenvalue pdf of $\boldsymbol{F}$ in (\ref{eqn:f}) for large but finite $m_1,m_2,$ and $q$ is given by 
	\begin{align}
		\bar{p}_{\lambda(\mathfrak{F})}(\lambda;{\rho_1,\rho_2}) &= \begin{cases}
		\frac{(1-\rho_1)\sqrt{(\lambda-l_f)(u_f - \lambda)}}{2 \pi \rho_1 \lambda (\lambda +1)} & \text{if $l_f \leq \lambda \leq u_f$}\\
		0 & \text{otherwise,}
		\end{cases} \label{eqn:p_asympt}
	\end{align}
	where $l_f = \frac{\rho_1}{\rho_2}\frac{1-\sqrt{1 - (1-\rho_1)(1-\rho_2)}}{(1-\rho_1)^2}$ and $u_f = \frac{\rho_1}{\rho_2}\frac{1+\sqrt{1 - (1-\rho_1)(1-\rho_2)}}{(1-\rho_1)^2}$ denote the limits of the support. Next, let $\boldsymbol{W}$ be defined as in Theorem \ref{thm:e}. Furthermore, let $\xi = \frac{q}{p}, p > q.$ Then, based on \cite[Thm. 2.35]{Tulino2004}, an approximation of the marginal eigenvalue pdf of $\boldsymbol{W}$ in (\ref{eqn:e}) for large but finite $p$ and $q$ is given by
	\begin{align}
		\bar{p}_{\lambda(\mathfrak{W})}(\lambda; \xi) &= \begin{cases}
		\frac{\sqrt{(x-l_w) (u_w-x)}}{2 \pi\xi\lambda p} & \text{if $l_w \leq \lambda \leq u_w$}\\
		0 & \text{otherwise,}
		\end{cases} \label{eqn:w_asympt}
	\end{align}
	where $l_w = p (1-\sqrt{\xi})^2$ and $u_w = p (1+\sqrt{\xi})^2$ denote the limits of the support. Approximate ergodic rate expressions $\bar{R}_1$ and $\bar{R}_2$ based on (\ref{eqn:r1n1})-(\ref{eqn:r2n}) can be obtained by replacing the densities $p_{\lambda(\mathfrak{F})}(\lambda; m_1,m_2,q)$ and $p_{\lambda(\mathfrak{W})}(\lambda; p,q)$ by $\bar{p}_{\lambda(\mathfrak{F})}(\lambda;{\rho_1,\rho_2})$ and $\bar{p}_{\lambda(\mathfrak{W})}(\lambda; \xi),$ respectively.
	
	\begin{remark}
	The pdf $\bar{p}_{\lambda(\mathfrak{F})}(\lambda; {\rho_1, \rho_2})$ does not exist for the case $\rho_1=1,$ which is an important case for MIMO-NOMA, as $l_f$ and $u_f$ are infinity\footnote{This is because, in the asymptotic regime, for $\rho_1=1,$ $\boldsymbol{F}$ entails the inversion of an asymptotic square matrix resulting, with probability $1,$ in eigenvalues which are $+\infty.$}. 
	\end{remark}
	
	\subsection{Unequal Power Allocation}
	\label{sec:upa}
	In this section, we derive simplified ergodic rate expressions for UPA based on finite-size RMT. In UPA, the powers allocated to the SISO-NOMA symbols $s_l, l=1,\dots,M,$ need not be identical. However, identical powers are allocated to the remaining user symbols\footnote{Although we restrict ourselves to non-equal powers only for the SISO-NOMA symbols, the techniques presented in this section can be utilized to extend the analysis to non-equal powers for all symbols in a straightforward manner.}, as described in the following. For symbols $s_l, l=1,\dots,M,$ transmit powers $p_{1,l}$ and $p_{2,l}$ are allocated to users 1 and 2, respectively. Furthermore, the $\bar{M}_1$ symbols $s_l,l=M+1,\dots,M+\bar{M}_1,$ are allocated a transmit power of $p_1.$ Lastly, the $\bar{M}_2$ symbols $s_l,l=M+\bar{M}_1,\dots,L,$ are allocated a transmit power of $p_2.$
			
	\begin{proposition}
		\label{prop:pwane}
		For UPA, (\ref{eqn:eppleq1}) simplifies as follows:
		\begin{align}		
		P_T = \begin{cases}
		p_1 + p_2 & \text{if $M_1+M_2 \leq N$} \\
		\frac{1}{M_1+M_2-N} \sum_{l = 1}^{M} (p_{1,l} + p_{2,l}) + \frac{p_1\bar{M}_1 + p_2\bar{M}_2}{M_1+M_2-N} & \text{if $M_1 + M_2 > N$ and $M_1 \geq N$} \\
		\frac{M_1}{N M} \sum_{l = 1}^{M} (p_{1,l} + p_{2,l}) + p_1\frac{\bar{M}_1}{M} + p_2 & \text{otherwise.}
		\end{cases} \label{eqn:ptupa}
		\end{align}
	\end{proposition}
	\begin{proof}
	For $M_1+M_2 \leq N,$ based on (\ref{eqn:bdpt}), (\ref{eqn:eppleq1}) simplifies to $P_T = p_1 + p_2.$ For $M_1+M_2 > N,$ analogous to Proposition \ref{prop:pwae}, (\ref{eqn:eppleq1}) simplifies to (\ref{eqn:ptupa}) based on \cite[Thm. 3.3.8]{Nagar2011}.
	\end{proof}

		Unlike for EPA, for UPA, the ergodic achievable rates of the users cannot be expressed in an additively separable form utilizing a common function $g(\cdot)$ for all $[\boldsymbol{\Sigma}_k]_{ll}, k=1,2, l=1,\dots,M,$ as required in (\ref{eqn:addsep}), due to the flexible power allocation across the symbols. Hence, for UPA, Definition \ref{def:mepdf} cannot be exploited for simplification of the ergodic achievable rates. Instead, the ergodic achievable rate expressions can be simplified and the dimension for integration can be reduced by utilizing the pdfs of the ordered values $[\boldsymbol{\Sigma}_k]_{ll}, k=1,2, l=1,\dots,M,$ as described in the following. Furthermore, the ordered pdfs can also be exploited for long-term power allocation, see Section \ref{sec:pa}.

	
	\begin{definition}[Based on \cite{Zanella2008, Zanella2009}]
		\label{def:oepdf}
		Let $\boldsymbol{X}$ be a $q\times q$ Hermitian symmetric random matrix with eigenvalues $\lambda_1 \geq \lambda_2 \geq \dots \geq \lambda_q,$ and let $\boldsymbol{\lambda} = \{\lambda_n,n=1,\dots,q\}.$ Furthermore, let the joint pdf of the eigenvalues be denoted by $p_{\boldsymbol{\lambda}}(\boldsymbol{\lambda}).$ Then, the $l$-th ordered eigenvalue pdf is defined as
		\begin{align}
		p_l(\lambda_l) = \int_{\lambda_l}^{+\infty}\dots\int_{\lambda_2}^{+\infty} \left[\int_{0}^{\lambda_l}\dots\int_{0}^{\lambda_{q-1}}  p_{\boldsymbol{\lambda}}(\boldsymbol{\lambda}) \mathrm{d}\lambda_q\cdots\mathrm{d}\lambda_{l+1}\right] \mathrm{d}\lambda_1\cdots\mathrm{d}\lambda_{l-1}. \label{eqn:defoe}
		\end{align}
	\end{definition}

	Using the above definition, for a $q\times q$ Hermitian symmetric random matrix $\boldsymbol{X}$ with eigenvalues $\lambda_1, \lambda_2, \dots, \lambda_q,$ we may simplify the expectation of a generic additively separable function
	\begin{align}
	f(\lambda_1,\dots,\lambda_q) = \sum_{l = 1}^{q} g_l(\lambda_l)
	\end{align}
	in the eigenvalues of $\boldsymbol{X}$ as
	\begin{align}
	\mathrm{E}\left[f(\lambda_1,\dots,\lambda_q)\right] = \mathrm{E}\left[\sum_{l=1}^{q} g_l(\lambda_l)\right] \overset{(a)}{=} \sum_{l=1}^{q} \int_{0}^{+\infty} g_l(\lambda_l) p_l(\lambda_l) \mathrm{d}\lambda_l, \label{eqn:opdfe}
	\end{align}
	where (a) is obtained using Definition \ref{def:oepdf}. In the following, (\ref{eqn:opdfe}) is used to simplify the achievable ergodic rates of the users in terms of the ordered eigenvalue pdfs for UPA.

		To obtain the ordered pdfs of $[\boldsymbol{\Sigma}_k]_{ll}, k=1,2, l=1,\dots,M,$ based on Proposition \ref{prop:equif}, the ordered eigenvalue pdfs of $\boldsymbol{F},$ denoted by $p_l(\lambda_l; \mu_1, \mu_2, \nu),$ can be exploited. Furthermore, the obtained pdfs can be used to simplify the expressions for $R_{1,l}$ and $R_{2,l}, l=1,\dots,M,$ via (\ref{eqn:opdfe}). $p_l(\lambda_l; \mu_1, \mu_2, \nu)$ is provided in the following theorem.

	\begin{theorem}
	\label{thm:f2}
		Let $\boldsymbol{X} \sim \mathcal{CW}_q(m_1,\boldsymbol{I}_q)$ and $\boldsymbol{Y} \sim \mathcal{CW}_q(m_2,\boldsymbol{I}_q)$ be $q\times q$ independent complex-valued Wishart random matrices with $m_1,m_2 \geq q$ degrees of freedom. The pdf of the $l$-th ordered eigenvalue of matrix $\boldsymbol{F} = \boldsymbol{Y}^{\frac{1}{2}}\boldsymbol{X}^{-1}\boldsymbol{Y}^{\frac{1}{2}}$ is given by
		\begin{align}
		p_l(\lambda_l; m_1, m_2, q) = K_{p_l} g^{[l,(),()]}_l(\lambda_l; m_1, m_2, q),
		\end{align}
		where $K_{p_l}$ is a constant ensuring that the integral over the pdf is equal to one, and function $g^{[d,\boldsymbol{n},\boldsymbol{m}]}_l(\lambda_l; m_1, m_2, q)$ is given by the recurrence relation
		\begin{align}
		\sum_{n = 1}^{|{\mathcal{I}}^{[d,\boldsymbol{n}]}|} \sum_{m = 1}^{|{\mathcal{I}}^{[d,\boldsymbol{m}]}|} g^{[d-1,\boldsymbol{n}',\boldsymbol{m}']}_l(\lambda_l; m_1,m_2,q), \label{eqn:rec}
		\end{align}
		$[l,(),()]$ denotes the initial value of $[d,\boldsymbol{n},\boldsymbol{m}],$ and ``()'' denotes the empty tuple. Tuples $\boldsymbol{n}$ and $\boldsymbol{m}$ are updated as $\boldsymbol{n}' \coloneqq \boldsymbol{n} \cup \{(n,[{\mathcal{I}}^{[d,\boldsymbol{n}]}]_n)\}$ and $\boldsymbol{m}' \coloneqq \boldsymbol{m} \cup \{(m,[{\mathcal{I}}^{[d,\boldsymbol{m}]}]_m)\},$ where $n$ and $m$ are the summation indices in (\ref{eqn:rec}), and $[{\mathcal{I}}^{[d,\boldsymbol{n}]}]_n$ and $[{\mathcal{I}}^{[d,\boldsymbol{m}]}]_m$ are the $n$-th and $m$-th elements of sets ${\mathcal{I}}^{[d, \boldsymbol{n}]}$ and ${\mathcal{I}}^{[d, \boldsymbol{m}]},$ respectively, defined as ${\mathcal{I}}^{[d,\boldsymbol{n}]} \coloneqq \{1,2,\dots,q\} \setminus \pi_2(\boldsymbol{n})$ and ${\mathcal{I}}^{[d,\boldsymbol{m}]} \coloneqq \{1,2,\dots,q\} \setminus \pi_2(\boldsymbol{m}).$ Next, the termination step is
		\begin{align}
		\scalebox{0.8}{\mbox{\ensuremath{\displaystyle g^{[1,\boldsymbol{n},\boldsymbol{m}]}_l(\lambda_l; m_1, m_2, q)}}} &= \scalebox{0.75}{\mbox{\ensuremath{\displaystyle \sum_{n = 1}^{|{\mathcal{I}}^{[d,\boldsymbol{n}]}|} \sum_{m = 1}^{|{\mathcal{I}}^{[d,\boldsymbol{m}]}|} s\left(\boldsymbol{n}',\boldsymbol{m}'\right) \frac{\lambda_l^{n+m-2+m_2-q}}{(1+\lambda_l)^{m_1+m_2}}  \mathrm{det}\left(\boldsymbol{\Xi}\left(l, m_1, m_2, q, {\mathcal{I}}^{[d+1,\boldsymbol{n}']},{\mathcal{I}}^{[d+1,\boldsymbol{m}']}\right)\right)}}} \nonumber\\
		&\qquad\scalebox{0.75}{\mbox{\ensuremath{\displaystyle \times \prod_{i = 1}^{l-1}(-1)^{[{\mathcal{I}}^{[d,\boldsymbol{n}]}]_i{+}[{\mathcal{I}}^{[d,\boldsymbol{m}]}]_i{+}1{-}m_1{-}q}\: \mathrm{B}(-\lambda_l^{-1},1{+}m_1{-}[{\mathcal{I}}^{[d,\boldsymbol{n}]}]_i{-}[{\mathcal{I}}^{[d,\boldsymbol{m}]}]_i{+}q,1{-}m_1{-}m_2)}}} \label{eqn:f2term}
		\end{align}
		where $\mathrm{B}(\cdot,\cdot,\cdot)$ is the incomplete Beta function, $\boldsymbol{\Xi}\left(l, m_1, m_2, q, {\mathcal{I}}^{[d+1,\boldsymbol{n}']},{\mathcal{I}}^{[d+1,\boldsymbol{m}']}\right)$ is a $(q-l)\times (q-l)$ matrix with elements 
		\begin{align}
		\scalebox{0.7}{\mbox{\ensuremath{\displaystyle \left[\frac{\lambda_l^{(m_2-q+[{\mathcal{I}}^{[d+1,\boldsymbol{n}']}]_i+[{\mathcal{I}}^{[d+1,\boldsymbol{m}']}]_j-1)}\:{}_2F_1(m_1{+}m_2,m_2{-}q{+}[{\mathcal{I}}^{[d+1,\boldsymbol{n}']}]_i{+}[{\mathcal{I}}^{[d+1,\boldsymbol{m}']}]_j{-}1,m_2{-}q{+}[{\mathcal{I}}^{[d+1,\boldsymbol{n}']}]_i{+}[{\mathcal{I}}^{[d+1,\boldsymbol{m}']}]_j;-\lambda_l)}{m_2-q+[{\mathcal{I}}^{[d+1,\boldsymbol{n}']}]_i+[{\mathcal{I}}^{[d+1,\boldsymbol{m}']}]_j-1}\right]_{ij}}}}, \label{eqn:f2mat}
		\end{align}
		with $i,j=1,\dots,q-l,$
		\begin{align}
		s\left(\boldsymbol{n}',\boldsymbol{m}'\right) &= (-1)^{\sum_{i=1}^{|\boldsymbol{n}'|} [\pi_1(\boldsymbol{n}')]_i + [\pi_1(\boldsymbol{m}')]_i},
		\end{align}
		and ${}_2F_1(\cdot,\cdot\,;\cdot)$ denotes the Gaussian hypergeometric function \cite{Olver2010}.
	\end{theorem}
	\begin{proof}
		Please refer Appendix \ref{app:f2}.
	\end{proof}
	
	\begin{remark}
		Although the expressions given in Theorem \ref{thm:f2} for $p_l(\lambda_l; m_1, m_2, q)$ are cumbersome, the resulting pdfs are simple polynomial expressions. A few examples are given in Table \ref{tab:poly}.
	\end{remark}
	
	\begin{table}
		\centering
		\caption{Expressions for the pdf in Theorem \ref{thm:f2} for various $m_1, m_2,$ and $q.$}
		\label{tab:poly}
		\begin{tabular}{|l|c||l|c|}
			\hline
			$(m_1,m_2,q)$	& $p_l(\lambda_l;m_1,m_2,q)$ & $(m_1,m_2,q)$ & $p_l(\lambda_l;m_1,m_2,q)$ \\
			\hline
			\hline
			$(3,3,1), l=1$ & $\frac{30 \lambda_1}{(1+\lambda_1)^6}$  & $(4,4,4), l=1$  & 	$\rule{0pt}{12pt}\frac{16 \lambda_1^{15}}{(\lambda_1 + 1)^{17}}$ \\\hline
			$(4,1,1), l=1$ & $\frac{4}{(1+\lambda_1)^5}$ &  $(4,4,4), l=2$  &   $\rule{0pt}{13pt}\frac{16 \lambda_2^8 (100 \lambda_2^4 + 450 \lambda_2^3 + 828 \lambda_2^2 + 700 \lambda_2 + 225)}{(\lambda_2 + 1)^{17}}$ \\\hline
			$(3,3,2), l=1$ & $\frac{12 \lambda_1^5 (5 + 3\lambda_1) }{(1+\lambda_1)^9}$ & $(4,4,4), l=3$  &   $\rule{0pt}{13pt}\frac{16 \lambda_3^3 (225 \lambda_3^4 + 700 \lambda_3^3 + 828 \lambda_3^2 + 450 \lambda_3 + 100)}{(\lambda_3 + 1)^{17}}$  \\\hline
			$(3,3,2), l=2$ & $\frac{12 \lambda_2 (5 + 3\lambda_2) }{(1+\lambda_2)^9}$ & $(4,4,4), l=4$  &	$\rule{0pt}{10pt}\frac{16}{(\lambda_4 + 1)^{17}}$ \\\hline
		\end{tabular}
		\vspace{-0.5cm}
	\end{table}
	
	The ordered eigenvalue pdf and the marginal eigenvalue pdf are related as specified in the following corollary.
	\begin{corollary}
		\label{cor:ff2}
		The marginal eigenvalue pdf given in Theorem \ref{thm:f}, $p_{\lambda(\mathfrak{F})}(\lambda; m_1, m_2, q),$ and the ordered eigenvalue pdf given in Theorem \ref{thm:f2}, $p_l(\lambda_l; m_1, m_2, q),$ are related as follows:
		\begin{align}
			p_{\lambda(\mathfrak{F})}(\lambda; m_1, m_2, q) = \frac{1}{q} \sum_{l=1}^{q} p_l(\lambda; m_1, m_2, q).
		\end{align}
	\end{corollary}
	\begin{proof}
		The proof follows directly from Definitions \ref{def:mepdf} and \ref{def:oepdf}.
	\end{proof}
	
	Next, we show that based on the ordered eigenvalue pdf for $\lambda_l$ given in Theorem \ref{thm:f2}, the expressions in (\ref{eqn:r}) can be simplified.
	
	\begin{proposition}
	\label{prop:neq}
	For UPA and $M_1+M_2\leq N,$ the ergodic rates of users 1 and 2 for the proposed UA-SD MIMO-NOMA scheme are given by
	\begin{align}
	\scalebox{0.9}{\mbox{\ensuremath{\displaystyle R_1}}} &= \scalebox{0.9}{\mbox{\ensuremath{\displaystyle \bar{M}_1\int_{0}^{+\infty} \log_2\left(1+\frac{p_1\lambda}{\Pi_1\sigma^2}\right) p_{\lambda(\mathfrak{W})}(\lambda; {M_1,\bar{M}_1}) \mathrm{d}\lambda}}}, \label{eqn:or1n1}\\
	\scalebox{0.9}{\mbox{\ensuremath{\displaystyle R_2}}} &= \scalebox{0.9}{\mbox{\ensuremath{\displaystyle \bar{M}_2\int_{0}^{+\infty} \log_2\left(1+\frac{p_2\lambda}{\Pi_2\sigma^2}\right) p_{\lambda(\mathfrak{W})}(\lambda; {M_2,\bar{M}_2}) \mathrm{d}\lambda}}}, \label{eqn:or2n1}
	\end{align}
	and for UPA and $M_1+M_2 > N,$ the ergodic rates are given by
	\begin{align}
	\scalebox{0.9}{\mbox{\ensuremath{\displaystyle R_1}}} &= \scalebox{0.9}{\mbox{\ensuremath{\displaystyle \sum_{l=1}^{M}\int_{0}^{\Pi^{-1}} \log_2\left(1+\frac{1}{\Pi_2}\frac{\frac{\lambda_l p_{1,l}}{1+\lambda_l}}{\sigma^2+\frac{1}{\Pi_2}\frac{\lambda_l p_{2,l}}{1+\lambda_l}}\right) p_l(\lambda_l; {\mu_1,\mu_2,\nu}) \mathrm{d}\lambda_l}}} \nonumber\\
	&\qquad + \scalebox{0.9}{\mbox{\ensuremath{\displaystyle \sum_{l=1}^{M} \int_{\Pi^{-1}}^{+\infty} \log_2\left(1+\frac{1}{\Pi_1}\frac{\frac{p_{1,l}}{1+\lambda_l}}{\sigma^2+\frac{1}{\Pi_1}\frac{p_{2,l}}{1+\lambda_l}}\right) p_l(\lambda_l; {\mu_1,\mu_2,\nu}) \mathrm{d}\lambda_l +  \bar{M}_1 \log_2\left(1+\frac{p_1}{\Pi_1\sigma^2}\right)}}}, \label{eqn:or1n}\\
	\scalebox{0.9}{\mbox{\ensuremath{\displaystyle R_2}}} &= \scalebox{0.9}{\mbox{\ensuremath{\displaystyle \sum_{l=1}^{M} \int_{0}^{+\infty} \log_2\left(1+\frac{1}{\Pi_2}\frac{\lambda_l p_{2,l}}{(1+\lambda_l)\sigma^2}\right) p_l(\lambda_l; {\mu_1,\mu_2,\nu}) \mathrm{d}\lambda_l}}}\nonumber\\
	&\hspace{3cm}\scalebox{0.9}{\mbox{\ensuremath{\displaystyle + \bar{M}_2\int_{0}^{+\infty} \log_2\left(1+\frac{p_2 \lambda}{\Pi_2\sigma^2}\right) p_{\lambda(\mathfrak{W})}(\lambda; {M_2,\bar{M}_2}) \mathrm{d}\lambda}}}, \label{eqn:or2n}
	\end{align}
	where $\mu_1,\mu_2,$ and $\nu$ are given in Table \ref{tab:map}, and $\Pi= \frac{\Pi_1}{\Pi_2}.$
	\end{proposition}
	\begin{proof}
	Analogous to the proof for Proposition \ref{prop:eq}, for UPA, (\ref{eqn:r}) can be simplified to (\ref{eqn:or1n1}) and (\ref{eqn:or2n1}) for $M_1+M_2 \leq N$ based on Proposition \ref{prop:d1d2} and Definition \ref{def:mepdf}. For $M_1+M_2 > N,$ (\ref{eqn:or1n}) and (\ref{eqn:or2n}) can be obtained based on (\ref{eqn:r12min}), Propositions \ref{prop:equif} and \ref{prop:d1d2}, and Definition \ref{def:oepdf}.
	\end{proof}

	\begin{remark}
	From Corollary \ref{cor:ff2}, we note that the proposed UA-SD MIMO-NOMA scheme with EPA could also be analyzed using the ordered eigenvalue pdfs. However, the ergodic rate expressions based on the marginal eigenvalue pdf given in (\ref{eqn:r1n}) and (\ref{eqn:r2n}) can be computed more efficiently than the corresponding expressions based on the ordered eigenvalue pdfs given in (\ref{eqn:or1n}) and (\ref{eqn:or2n}).
	\end{remark}

	\begin{remark}
	
		Due to the analytical complexity of the pdfs of the marginal and ordered eigenvalues for general $M_1, M_2,$ and $N,$ closed-form solutions for the ergodic achievable rates in Propositions \ref{prop:eq} and \ref{prop:neq} cannot be obtained. However, for given $M_1, M_2,$ and $N,$ the pdfs in Theorems \ref{thm:f} and \ref{thm:f2} can be simplified to polynomials, see Table \ref{tab:poly}. Hence, evaluating the ergodic achievable rate requires only simple one-dimensional integration. The associated computational complexity is much lower than that of evaluating the ergodic achievable rates via time-consuming Monte-Carlo simulations.
	
	\end{remark}
	\subsection{Power Allocation Algorithm}
	\label{sec:pa}
	In this section, to limit the additional complexity introduced by power allocation, we develop a long-term power allocation algorithm for maximization of the weighted ergodic sum rate, which depends only on the channel statistics.
	
	First, based on (\ref{eqn:or1n1})-(\ref{eqn:or2n}), we formulate optimization problem P1 for maximization of the weighted ergodic sum rate as follows:
	\begin{align}
		\text{ P1: } \max_{p_1,p_2,p_{k,l} \geq 0\,\forall\,k,l} \eta R_1 + (1-\eta) R_2 \quad \text{ s.t. } \quad P_T \leq P_\mathrm{max},
	\end{align}
	where $\eta \in [0,1]$ is a fixed weight which can be chosen to adjust the rates of user 1 and 2 \cite[Sec. 4]{WangGiannakis2011}, and $P_\mathrm{max}$ is the available transmit power budget. For $M_1+M_2 \leq N,$ problem P1 is convex. However, for $M_1+M_2 > N,$ problem P1 is non-convex, due to the coupling between $p_{1,l}$ and $p_{2,l}, l=1,\dots,M,$ in $R_1$ given in (\ref{eqn:or1n}). Hence, efficient convex optimization methods cannot be used for obtaining the global maximum. Instead, for $M_1+M_2 > N,$ we utilize the low-complexity concave-convex procedure (CCP) \cite{Yuille2003}, \cite{Lipp2016} to obtain a suboptimal solution.
	
	To this end, we replace $R_1$ by a concave overestimator, denoted by $\tilde{R}_1,$ obtained via a first-order approximation of the non-concave terms around $p_{2,l}=q_l,$ $l=1\dots,M,$ given (excluding constant terms) by
	\begin{align}
		\scalebox{0.8}{\mbox{\ensuremath{\displaystyle \tilde{R}_1}}} &= \scalebox{0.8}{\mbox{\ensuremath{\displaystyle \sum_{l=1}^{M}\int_{0}^{\Pi^{-1}} \log_2\left(1+\frac{\lambda_l p_{1,l}}{\Pi_2 \sigma^2(1+\lambda_l)} + \frac{\lambda_l p_{2,l}}{\Pi_2 \sigma^2 (1+\lambda_l)}\right) p_l(\lambda_l; {\mu_1,\mu_2,\nu}) \mathrm{d}\lambda_l}}} \nonumber\\ &\quad\scalebox{0.8}{\mbox{\ensuremath{\displaystyle + \sum_{l=1}^{M} \int_{\Pi^{-1}}^{+\infty} \log_2\left(1+ \frac{p_{1,l}}{\Pi_1 \sigma^2 (1+\lambda_l)} + \frac{p_{2,l}}{\Pi_1 \sigma^2(1+\lambda_l)}\right) p_l(\lambda_l; {\mu_1,\mu_2,\nu}) \mathrm{d}\lambda_l +  L_1 + \bar{M}_1 \log_2\left(1+\frac{p_1}{\Pi_1\sigma^2}\right)}}}, \label{eqn:r1a} \\
		\scalebox{0.75}{\mbox{\ensuremath{\displaystyle L_1}}} &= \scalebox{0.75}{\mbox{\ensuremath{\displaystyle -\frac{1}{\log(2)}\sum_{l=1}^{M}\left[\int_{0}^{\Pi^{-1}}\frac{\lambda_l(p_{2,l}-q_l)}{\lambda_l q_{l} + \Pi_2 \sigma^2 (1+\lambda_l)} p_l(\lambda_l; {\mu_1,\mu_2,\nu}) \mathrm{d}\lambda_l + \int_{\Pi^{-1}}^{+\infty} \frac{(p_{2,l}-q_l)}{q_{l} + \Pi_1 \sigma^2 (1+\lambda_l)} p_l(\lambda_l; {\mu_1,\mu_2,\nu}) \mathrm{d}\lambda_l\right].}}}
	\end{align}
	Next, we construct a concave optimization problem P2, based on $\tilde{R}_1,$ as follows:
	\begin{align}
		\text{P2: } \max_{p_1,p_2,p_{k,l} \geq 0\,\forall\,k,l} \eta \tilde{R}_1 + (1-\eta) R_2 \quad \text{ s.t. } \quad \quad P_T \leq P_\mathrm{max},
	\end{align}
	which is solved iteratively. First, we initialize $q_l^{(0)} = 0,$ $l=1,\dots,M.$ Next, in iteration $n=1,2,\dots,$ we solve P2 given $q_l^{(n-1)}$ to obtain the optimal solution $p_1^{(n)},p_2^{(n)},p_{k,l}^{(n)},k=1,2,l=1,\dots,M.$ Then, we update $q_l^{(n)} = p_{2,l}^{(n)},l=1,\dots,M,$ to obtain a tighter concave overestimator in the next iteration. The iterations are continued until convergence, upto a suitable numerical tolerance $\epsilon.$ The CCP is guaranteed to converge to a stationary point of P1 \cite{Yuille2003}, \cite{Lipp2016}. The proposed algorithm is summarized in Algorithm \ref{alg:e1}.

	\begin{figure}
		\begin{algorithm}[H]
			\small
			\begin{algorithmic}[1]
				\STATE {Initialize $q_l^{(0)}=0,$ $p_1^{(0)}=p_2^{(0)}=p_{k,l}^{(0)}=-\infty$ for $k=1,2,l=1,\dots,M,$ numerical tolerance $\epsilon,$ and iteration index $n=0$}
				\REPEAT	
				\STATE {$n \leftarrow n + 1$}
				\STATE {Solve P2 given $q_l^{(n-1)},l=1,\dots,M,$ to obtain the solution $p_1^{(n)},p_2^{(n)},p_{k,l}^{(n)}\,\forall\,k,l$ \label{step:alg1_2}}
				\STATE {Update $q_l^{(n)}=p_{2,l}^{(n)},l=1,\dots,M,$ for the next iteration \label{step:alg1_3}}
				\UNTIL {$|p_1^{(n)} - p_1^{(n-1)}| < \epsilon,$ $|p_2^{(n)} - p_2^{(n-1)}| < \epsilon,$ and $|p_{k,l}^{(n)} - p_{k,l}^{(n-1)}| < \epsilon\,\forall\,k,l$}
				\STATE{Return $p_1^{(n)},p_2^{(n)},p_{k,l}^{(n)},k=1,2,l=1,\dots,M,$ as the power allocation}
			\end{algorithmic}
			\caption{Power Allocation Algorithm for UPA and $M_1+M_2 > N.$}
			\label{alg:e1}
		\end{algorithm}
		\vspace{-1cm}
	\end{figure}

	\begin{remark}
	
		The computational complexity of Algorithm \ref{alg:e1} for a given number of iterations, $N_\mathrm{iter},$ can be obtained as follows. In Step \ref{step:alg1_2}, P2, which involves $2L$ real-valued optimization variables and an integral in the objective function, is solved. Let $S$ be the number of points utilized for evaluating the integral. Then, the computational complexity for solving P2 via an exponential-cone solver with numerical tolerance $\delta$ is given by $\mathcal{O}\mkern-\medmuskip\left(SL\sqrt{L}\log(1/\delta)\right)$ \cite{Serrano}. Furthermore, Step \ref{step:alg1_3} entails a complexity of $\mathcal{O}\mkern-\medmuskip\left(L\right)$ for updating the gradient, leading to a total complexity of Algorithm \ref{alg:e1} of  $\mathcal{O}\mkern-\medmuskip\left(N_\mathrm{iter}\left(L+SL\sqrt{L}\log(1/\delta)\right)\right).$
	
	\end{remark}

	\section{Simulation Results}
	\label{sec:sim}
	In this section, we first verify the pdf expressions in Theorems \ref{thm:f} and \ref{thm:f2}. Then, we compare the ergodic achievable rate regions of the proposed UA-SD MIMO-NOMA for EPA and UPA with those of GSVD-based MIMO-NOMA \cite{Chen2019} and TDMA-based MIMO-OMA.
	
	\subsection{Probability Density Functions}
	Figure \ref{fig:g} compares the marginal eigenvalue pdf, $p_{\lambda(\mathfrak{F})}(\lambda; m_1, m_2, q),$ derived in Theorem \ref{thm:f} for $m_1=5,m_2=4,$ and $q=2,$ with the empirical pdf obtained via Monte Carlo simulation and the approximation based on the asymptotic pdf given in (\ref{eqn:p_asympt}). From the figure, we note that the derived analytical result is in perfect agreement with the numerical simulation. However, for the considered small finite values of $m_1,m_2,$ and $q,$ the approximation based on the asymptotic pdf does not result in an accurate fit.
	
	Analogously, Figure \ref{fig:gl} compares the ordered eigenvalue pdfs, $p_l(\lambda_l; m_1, m_2, n),$ $l=1,\dots,4,$ derived in Theorem \ref{thm:f2} for $m_1=7,m_2=8,$ and $q=4,$ with empirical pdfs obtained via Monte Carlo simulation. We observe also in this case that the derived analytical results are in excellent agreement with the numerical simulation, thereby validating our derived expressions.
	
	\begin{figure}
	\centering
	\begin{minipage}[t]{0.48\textwidth}
		\centering
		\includegraphics[width=0.9\textwidth]{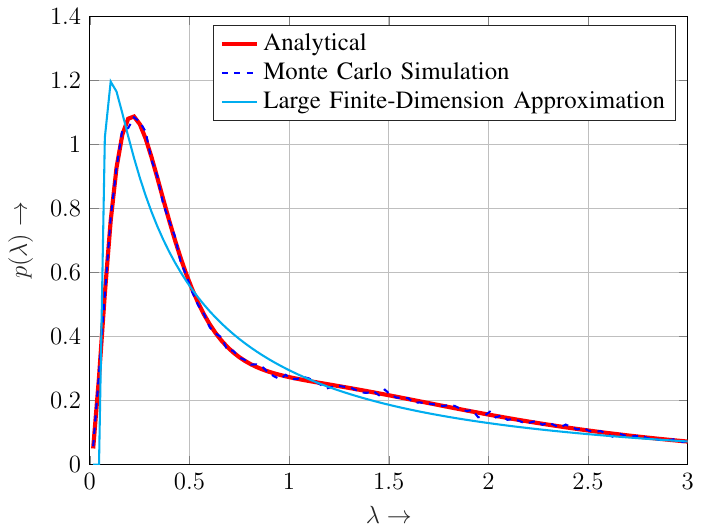}
		\caption{Simulation of the pdf given in Theorem \ref{thm:f} for $m_1=5,m_2=4,$ and $q=2.$}
		\label{fig:g}
	\end{minipage}%
	\hspace{0.04\textwidth}%
	\begin{minipage}[t]{0.48\textwidth}
		\includegraphics[width=0.9\textwidth]{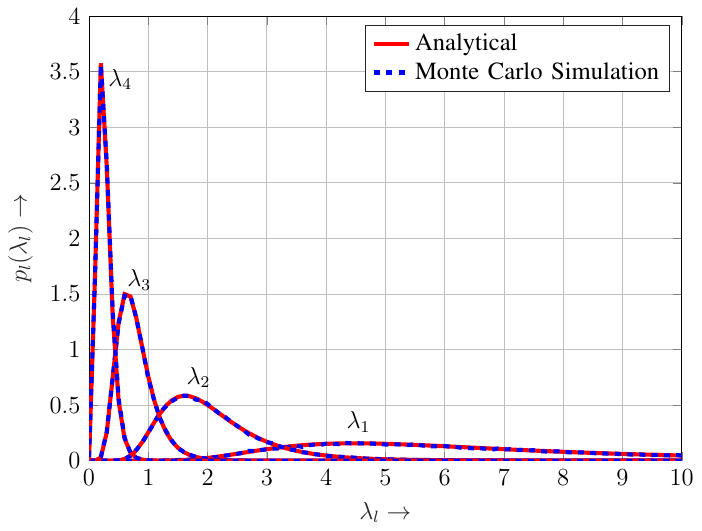}
		\caption{Simulation of the pdf given in Theorem \ref{thm:f2} for $m_1=7,m_2=8,$ and $q=4.$}
		\label{fig:gl}
	\end{minipage}
	\end{figure}
	
	\subsection{Ergodic Rate Regions}
	For evaluation of the ergodic rate regions, we assume that the first and second users are located at distances of $d_1 = 100~\text{m}$ and $d_2 = 10~\text{m}$ from the BS, respectively, as (MIMO-)NOMA is well suited for users with different channel conditions \cite{Saito2013}. The path loss is modeled as $\Pi_k = d_k^2,$ which corresponds to the free-space path-loss exponent of two \cite{Rappaport}. Furthermore, we assume a noise variance of $\sigma^2 = -35~\text{dBm}$ and set $P_T = P_\mathrm{max}$ for EPA. The maximum transmit power, $P_\mathrm{max},$ is chosen in the typical range of 4G systems, i.e., $10 \text{--} 30 \text{ dBm.}$ Moreover, the ergodic rate region for TDMA-based MIMO-OMA is obtained by time sharing of the SU-MIMO rates of the users. Lastly, for the performance results shown below, different values of $M_1, M_2,$ and $N$ are chosen to allow for a comprehensive comparison of UA-SD MIMO-NOMA with TDMA-based MIMO-OMA and GSVD-based MIMO-NOMA.
	
	Figure \ref{fig:335} shows the convex hull of the ergodic rate regions of the proposed UA-SD MIMO-NOMA for EPA and UPA, GSVD-based MIMO-NOMA \cite{Chen2019}, and TDMA-based MIMO-OMA for the case $\bar{M}_2 > 0$ with $M_1 = M_2 = 3, N = 5,$ and $P_\mathrm{max} = 20 \text{ dBm.}$ For EPA, the rate region is obtained by varying $P_1$ and $P_2$ in (\ref{eqn:r1n}) and (\ref{eqn:r2n}). For UPA, the rate region is obtained by solving P1 using Algorithm \ref{alg:e1} for different $\eta \in [0,1].$ For GSVD-based MIMO-NOMA, the rate region is obtained based on the asymptotic RMT expressions from \cite{Chen2019} which are modified based on our finite-size RMT framework.
	
	First, we note that, in Figure \ref{fig:335} (and in all subsequent figures), the ergodic rate regions obtained analytically via our finite-size RMT results and empirically via Monte Carlo simulation are in perfect agreement, thereby confirming the validity of our theoretical results.
		
	Furthermore, from Figure \ref{fig:335}, we observe that the proposed UA-SD MIMO-NOMA outperforms GSVD-based MIMO-NOMA for both EPA and UPA as channel inversion at the BS for $\bar{M}_2=2$ symbols, $s_{2,l},l=4,5,$ of user 2 is avoided. UA-SD MIMO-NOMA also outperforms TDMA-based MIMO-OMA for a wide range of user rates. However, TDMA-based MIMO-OMA is superior for the rate pairs close to the SU-MIMO rates due to the sustained channel inversion in the proposed scheme for the first $M+\bar{M}_1$ symbols, $s_1$ and $s_{1,l},l=2,3.$ Nevertheless, the proposed hybrid scheme, whose ergodic rate region is the convex hull of the union of the ergodic rate regions of UA-SD MIMO-NOMA and TDMA-based MIMO-OMA, as described in Section \ref{sec:hybrid}, outperforms TDMA-based MIMO-OMA for all possible rate pairs. Moreover, as expected, for the proposed UA-SD MIMO-NOMA, UPA outperforms EPA as UPA performs optimal power allocation. Lastly, the ergodic rate region obtained based on the asymptotic pdfs in (\ref{eqn:p_asympt}) and (\ref{eqn:w_asympt}) accurately approximates the ergodic rate region of the proposed UA-SD MIMO-NOMA with EPA.
	
	\begin{figure}
	\begin{minipage}[t]{0.48\textwidth}
		\centering
		\includegraphics[width=0.9\textwidth]{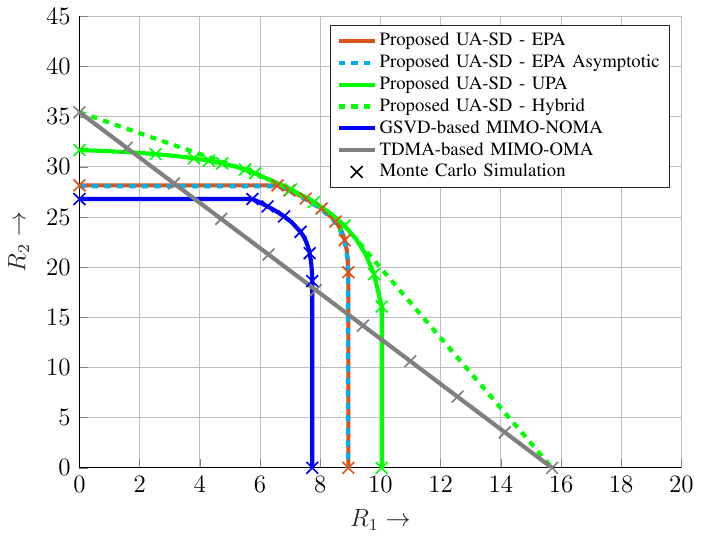}
		\caption{Ergodic rate region for $M_1=3, M_2=3, N=5,$ and $P_\mathrm{max} = 20 \text{ dBm}.$}
		\label{fig:335}
	\end{minipage}%
	\hspace{0.04\textwidth}%
	\begin{minipage}[t]{0.48\textwidth}
		\centering
		\includegraphics[width=0.9\textwidth]{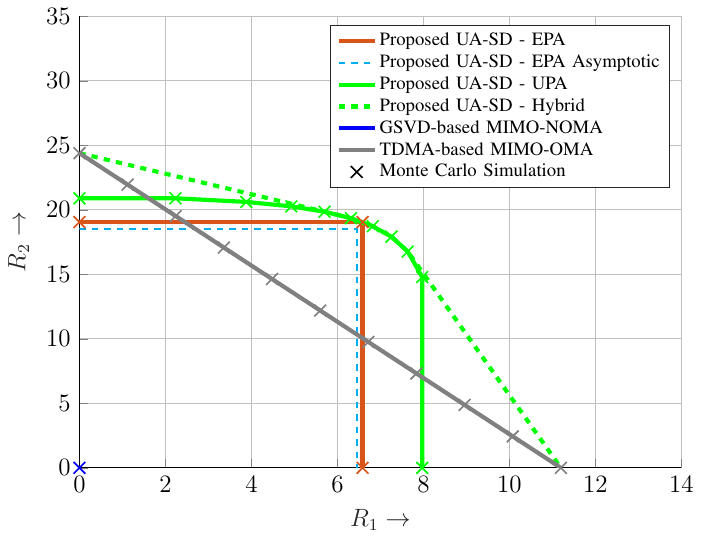}
		\caption{Ergodic rate region for $M_1=2, M_2=2, N=4,$ and $P_\mathrm{max} = 20 \text{ dBm}.$}
		\label{fig:224}
		\vspace{0.5cm}
	\end{minipage}

	\begin{minipage}[t]{0.48\textwidth}
		\centering
		\includegraphics[width=0.9\textwidth]{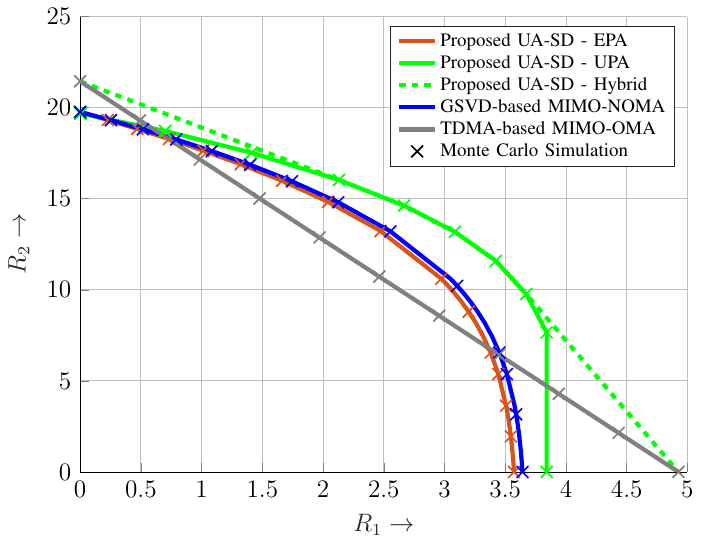}
		\caption{Ergodic rate region for $M_1=3, M_2=3, N=3,$ and $P_\mathrm{max} = 10 \text{ dBm}.$}
		\label{fig:333}
	\end{minipage}%
	\hspace{0.04\textwidth}%
	\begin{minipage}[t]{0.48\textwidth}
		\centering
		\includegraphics[width=0.9\textwidth]{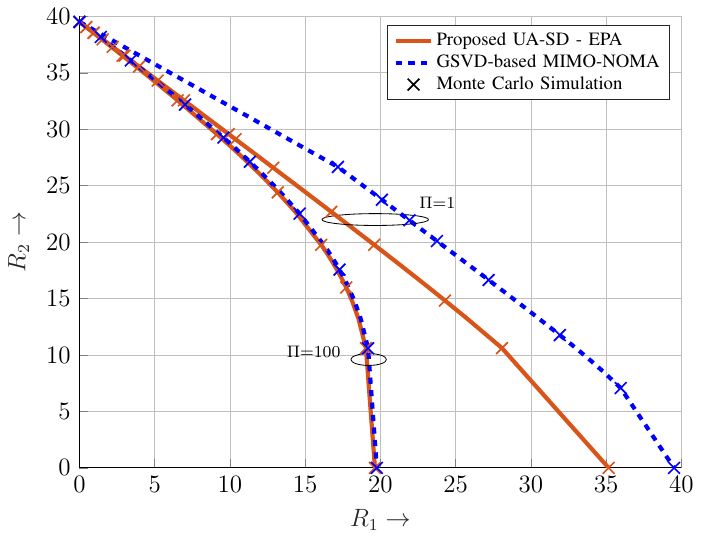}
		\caption{Ergodic rate region for $M_1=3, M_2=3, N=3,$ $P_\mathrm{max} = 30 \text{ dBm},$ and $\Pi = 1 \text{ and } 100.$}
		\label{fig:d333}
		\vspace{0.5cm}
	\end{minipage}

	\begin{minipage}[t]{0.48\textwidth}
		\centering
		\includegraphics[width=0.9\textwidth]{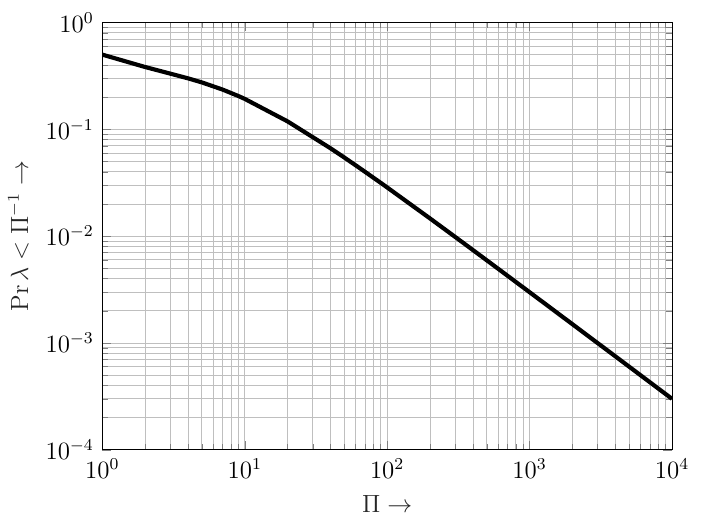}
		\caption{Plot of $\mathrm{Pr}\left\{\lambda < \Pi^{-1}\right\}$ for $M_1=3,$ $M_2=3,$ and $N=3$ as a function of $\Pi.$}
		\label{fig:p_degraded}
	\end{minipage}%
	\hspace{0.04\textwidth}%
	\begin{minipage}[t]{0.48\textwidth}
		\centering
		\includegraphics[width=0.9\textwidth]{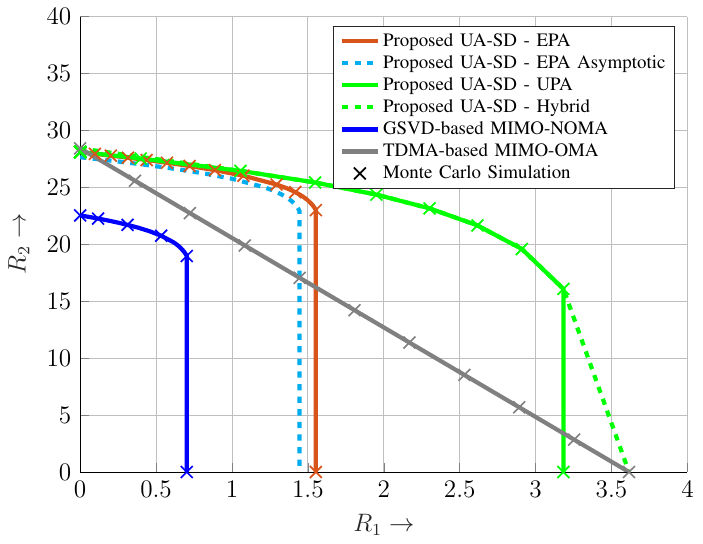}
		\caption{Ergodic rate region for $M_1=1, M_2=4,N=4,$ and $P_\mathrm{max} = 10 \text{ dBm}.$}
		\label{fig:144}
	\end{minipage}
	\end{figure}
	
	Figure \ref{fig:224} shows the convex hull of the ergodic rate regions of the same schemes as considered in Figure \ref{fig:335} for the case $M_1+M_2=N$ with $M_1 = M_2 = 2,N = 4,$ and $P_\mathrm{max} = 20 \text{ dBm.}$ As no symbols $s_l$ are transmitted to both users, for EPA, no rate adjustment between the users is possible. Furthermore, as $M_1+M_2=N,$ the achievable rate for GSVD-based MIMO-NOMA is zero as explained in Section \ref{sec:existing}. On the other hand, the proposed UA-SD MIMO-NOMA, which reduces to BD (i.e., spatial OMA) in this case, significantly outperforms TDMA-based MIMO-OMA for a wide range of user rates, and the proposed hybrid scheme outperforms TDMA-based MIMO-OMA over the entire rate region. Moreover, unlike EPA, UPA allows the adjustment of the user rates. Lastly, for EPA, the ergodic rate region obtained based on the asymptotic pdfs approximates the ergodic rate region of the proposed UA-SD MIMO-NOMA upto a small gap.
	
	Figure \ref{fig:333} shows the ergodic rate region for the case $\bar{M}_2 = 0$ with $M_1 = M_2 = N = 3,$ and $P_\mathrm{max} = 10 \text{ dBm},$ where all spatial streams are shared by the two users via SISO-NOMA. In this case, the proposed UA-SD MIMO-NOMA with UPA outperforms GSVD-based MIMO-NOMA. Furthermore, the proposed scheme with UPA outperforms TDMA-based MIMO-OMA for a wide range of user rates, and the proposed hybrid scheme outperforms TDMA-based MIMO-OMA over the entire rate region. Lastly, for the considered case, GSVD-based MIMO-NOMA yields a marginally larger rate region than UA-SD MIMO-NOMA with EPA, which restricts SIC to user 2 limiting the achievable ergodic rate of user 1, $R_{1,l},l=1,\dots,M,$ to the inferior rate at user 2, $R_{1,l}^{(2)},l=1,\dots,M,$ cf. Section \ref{sec:epa}. In the following, this effect is investigated in detail.

	Figure \ref{fig:d333} compares the rate regions of the proposed UA-SD MIMO-NOMA with EPA, which restricts SIC to user 2, and GSVD-based MIMO-NOMA, which allows SIC at users 1 and 2, cf. Section \ref{sec:decoding}, for $P_\mathrm{max} = 30 \text{ dBm},$ and $\Pi=1$ and $100.$ 	From the figure, we observe that the significant performance loss for $\Pi=1$ reduces to a negligible performance loss for $\Pi=100,$ thereby suggesting that for large $\Pi,$ the potential performance benefit from allowing SIC at users 1 and 2, as in GSVD-based MIMO-NOMA, is insignificant. Furthermore, as $\Pi$ increases, the potential performance gain decreases rapidly as it is proportional to $\mathrm{Pr}\left\{\lambda < \Pi^{-1}\right\},$ cf. Section \ref{sec:epa}, which also falls rapidly as $\Pi$ increases, as shown in Figure \ref{fig:p_degraded}.

	Figure \ref{fig:144} shows the convex hull of the ergodic rate region for the case $M_1 \neq M_2, \bar{M}_2 > 0,$ with $M_1 = 1,$ $M_2 = N = 4,$ and $P_\mathrm{max} = 10 \text{ dBm}.$ In this case, we observe that the proposed UA-SD MIMO-NOMA with EPA and UPA outperform GSVD-based MIMO-NOMA as they avoid channel inversion, cf. Section \ref{sec:gsvd}. In fact, even TDMA-based MIMO-OMA outperforms GSVD-based MIMO-NOMA as, for the considered case, $|M_1+M_2-N| = 1$ is small, leading to a poor performance of GSVD-based MIMO-NOMA, see Proposition \ref{prop:tgsvd}. Furthermore, for the proposed UA-SD MIMO-NOMA, UPA significantly outperforms EPA. Moreover, for EPA, the rate region obtained based on the asymptotic pdf has a small gap to the exact result as $M_1=1$ is small.

	\section{Conclusion}
	\label{sec:con}
	We proposed a UA-SD MIMO-NOMA scheme based on a new matrix decomposition that achieves SD through a combination of precoder design and low-complexity self-interference cancellation at the users, thereby lowering the decoding complexity at the users compared to joint decoding. Furthermore, we derived ergodic rate expressions for the proposed scheme for EPA and UPA based on a finite-size RMT framework. We exploited the derived ergodic rate expressions to develop a long-term power allocation algorithm for the proposed scheme which only depends on the channel statistics. Our performance comparisons based on the ergodic achievable rate regions revealed that the proposed scheme with EPA and UPA outperforms GSVD-based MIMO-NOMA \cite{Chen2019} and TDMA-based MIMO-OMA for most user rate combinations by avoiding channel inversion at the transmitter. Furthermore, a hybrid scheme employing time sharing enhanced the performance even further. Lastly, as in SISO-NOMA \cite{Saito2013}, the benefits of the proposed scheme are fully exploited when the ratio of the path loss coefficients of the users is large.
		
	Finally, we note that the pdfs derived in Theorems \ref{thm:f}, \ref{thm:e}, and \ref{thm:f2} can be exploited for the analysis of the diversity order, outage capacity, bit error rates, and outage probabilities of UA-SD MIMO-NOMA, which constitutes an interesting topic for future research.

	\begin{appendices}
	\renewcommand{\thesection}{\Alph{section}}
	\renewcommand{\thesubsection}{\thesection.\arabic{subsection}}
	\renewcommand{\thesectiondis}[2]{\Alph{section}:}
	\renewcommand{\thesubsectiondis}{\thesection.\arabic{subsection}:}
	\section{Proofs}
	\label{app:proofs}
	
	\subsection{Proof of Proposition \ref{prop:tgsvd}}
	\label{app:tgsvd}
	For $M_1 + M_2 > N,$ from (\ref{eqn:gsvd}), noting that, in this case, $\boldsymbol{Z}$ is invertible, we have
	\begin{align}
		\scalebox{1}{\mbox{\ensuremath{\displaystyle \boldsymbol{H}_1^\mathrm{H} \boldsymbol{H}_1 + \boldsymbol{H}_2^\mathrm{H} \boldsymbol{H}_2 = (\boldsymbol{Z}^{-1})^\mathrm{H} \left(\boldsymbol{C}^\mathrm{H} \boldsymbol{C} + \boldsymbol{S}^\mathrm{H} \boldsymbol{S}\right) \boldsymbol{Z}^{-1} = (\boldsymbol{Z}^{-1})^\mathrm{H} \boldsymbol{Z}^{-1} = (\boldsymbol{Z} \boldsymbol{Z}^\mathrm{H})^{-1}.}}}
	\end{align}
	Hence, (\ref{eqn:eppleq1}) can be simplified to
	\begin{align}
		\scalebox{0.8}{\mbox{\ensuremath{\displaystyle P_T}}} &\overset{(a)}{=} \scalebox{0.8}{\mbox{\ensuremath{\displaystyle P\mathrm{tr}\left(\mathrm{E}\left[\boldsymbol{Z} \boldsymbol{Z}^\mathrm{H}\right]\right) = P\mathrm{tr}\left(\mathrm{E}\left[\underbrace{\left(\boldsymbol{H}_1^\mathrm{H} \boldsymbol{H}_1 + \boldsymbol{H}_2^\mathrm{H} \boldsymbol{H}_2\right)^{-1}}_{\coloneqq \hat{\boldsymbol{W}}}\right]\right)}}} = \scalebox{0.8}{\mbox{\ensuremath{\displaystyle P\mathrm{tr}\left(\mathrm{E}\left[\hat{\boldsymbol{W}}\right]\right)}}},
	\end{align}
	where (a) is obtained by setting $\boldsymbol{P} = \boldsymbol{Z}$ in (\ref{eqn:eppleq1}), noting that $p_{1,l}+p_{2,l}=P,l=1,\dots,L,$ and exploiting $\mathrm{E}\left[\mathrm{tr}\left(\boldsymbol{Z}^\mathrm{H} \boldsymbol{Z}\right)\right] = \mathrm{tr}\left(\mathrm{E}\left[\boldsymbol{Z} \boldsymbol{Z}^\mathrm{H}\right]\right).$ Matrix $\hat{\boldsymbol{W}} \sim \mathcal{CW}_{N}(M_1+M_2,\boldsymbol{I}_{N})$ \cite[Thm. 3.3.8]{Nagar2011}. By applying \cite[Lemma 3.2 (ii)]{Nagar2011}, we obtain $P_T = \frac{PL}{N-(M_1+M_2)}.$ Proceeding analogously for $M_1+M_2\leq N$ and combining the results, we obtain the expression in (\ref{eqn:pgsvd}). \qed

	\subsection{Proof of Theorem \ref{thm:mgsvd}}
	\label{app:mgsvd}
	For $M_1+M_2 \leq N,$ matrices $\boldsymbol{Q}_1, \boldsymbol{Q}_2,$ and $\boldsymbol{Z}$ are obtained using BD as described in Section \ref{sec:bd}, \cite[Sec. III]{Spencer2004}.
	
	For $M_1+M_2 > N,$ we have $L=N.$ Let $\boldsymbol{K} \in \mathbb{C}^{N\times M}$ denote a matrix containing the basis vectors of the $M$ dimensional space $\mathrm{null}\left(\begin{bmatrix}\bar{\boldsymbol{H}}_1 &\bar{\boldsymbol{H}}_2\end{bmatrix}^\mathrm{H}\right) \cap \mathrm{col}\big(\boldsymbol{H}_1^\mathrm{H}\big) \cap \mathrm{col}\big(\boldsymbol{H}_2^\mathrm{H}\big)$, which exists when $M > 0.$ For the first user, matrix $\tilde{\boldsymbol{H}}_1 = \boldsymbol{H}_1\begin{bmatrix}\boldsymbol{K} & \bar{\boldsymbol{H}}_2\end{bmatrix}$ is forced to the identity matrix by using SVD. Let $\tilde{\boldsymbol{H}}_1 = \hat{\boldsymbol{U}}_1\boldsymbol{\Sigma}_1\hat{\boldsymbol{V}}_1^\mathrm{H},$ where $\hat{\boldsymbol{U}}_1 \in \mathbb{C}^{M_1\times M_1}$ and $\hat{\boldsymbol{V}}_1 \in \mathbb{C}^{M_1\times M_1}$ are unitary matrices, and $\boldsymbol{\Sigma}_1 \in \mathbb{R}^{M_1\times M_1}$ is a diagonal matrix which has the singular values of $\tilde{\boldsymbol{H}}_1$ on its main diagonal. Note that $\hat{\boldsymbol{U}}_1^\mathrm{H} \tilde{\boldsymbol{H}}_1 (\hat{\boldsymbol{V}}_1\boldsymbol{\Sigma}_1^+) = \boldsymbol{I}_{M_1}.$
	
	Next, for the second user, matrix $\hat{\boldsymbol{H}}_2 = \boldsymbol{H}_2\bar{\boldsymbol{H}}_1$ is diagonalized using SVD. Let	$\hat{\boldsymbol{H}}_2 = \hat{\boldsymbol{U}}_2\boldsymbol{\Sigma}_2\hat{\boldsymbol{V}}_2^\mathrm{H},$ where $\hat{\boldsymbol{U}}_2 \in \mathbb{C}^{M_2\times M_2}$ and $\hat{\boldsymbol{V}}_2 \in \mathbb{C}^{\bar{M}_2\times \bar{M}_2}$ are unitary matrices, and $\boldsymbol{\Sigma}_2 \in \mathbb{R}^{M_2\times \bar{M}_2}$ is a diagonal matrix which has the singular values of $\hat{\boldsymbol{H}}_2$ on its main diagonal. As a result, $\hat{\boldsymbol{U}}_2^\mathrm{H}\hat{\boldsymbol{H}}_2\hat{\boldsymbol{V}}_2 = \boldsymbol{\Sigma}_2$ is a diagonal matrix.
	
	Furthermore, for the second user, matrix 
	\begin{align}
		\tilde{\boldsymbol{H}}_2 = \hat{\boldsymbol{U}}_2^\mathrm{H} \left(\boldsymbol{H}_2 \begin{bmatrix}\boldsymbol{K} & \bar{\boldsymbol{H}}_2\end{bmatrix}\right) (\hat{\boldsymbol{V}}_1\boldsymbol{\Sigma}_1^+) \label{eqn:th2}
	\end{align}
	is diagonalized in two steps. First, QR decomposition is used to zero-out the last $\bar{M}_1$ columns of $\tilde{\boldsymbol{H}}_2.$ Next, SVD is used to diagonalize the remaining columns.
	
	Let, by QR decomposition, $\tilde{\boldsymbol{H}}_2^\mathrm{H}  = \boldsymbol{Q}\boldsymbol{R},$ where $\boldsymbol{Q} \in \mathbb{C}^{M_1\times M_1}$ is a unitary matrix, and $\boldsymbol{R} \in \mathbb{C}^{M_1\times M_2}$ is a rank $M$ upper triangular matrix. If $M$ is zero, then $\boldsymbol{Q} = \boldsymbol{I}_{M_1}$ and $\boldsymbol{R} = \boldsymbol{0}.$ In either case, $\tilde{\boldsymbol{H}}_2\boldsymbol{Q}$ is a matrix with the last $\bar{M}_1$ columns equal to zero.
	
	Lastly, the $(M_1+M_2-N)\times M$ matrix obtained by taking the last $M_1+M_2-N$ rows and the first $M$ columns of $\tilde{\boldsymbol{H}}_2\boldsymbol{Q},$ denoted by $\boldsymbol{B}_3,$ is diagonalized applying again SVD. Let
	\begin{align}
		\boldsymbol{B}_3 = \hat{\boldsymbol{U}}_3\boldsymbol{\Sigma}\hat{\boldsymbol{V}}_3^\mathrm{H}, \label{eqn:b3}
	\end{align}
	where $\hat{\boldsymbol{U}}_3 \in \mathbb{C}^{M\times M}$ and $\hat{\boldsymbol{V}}_3 \in \mathbb{C}^{M\times M}$ are unitary matrices, and $\boldsymbol{\Sigma} \in \mathbb{R}^{M\times M}$ is a diagonal matrix which has the singular values of $\boldsymbol{B}_3$ on its main diagonal. The diagonal entries of $\boldsymbol{\Sigma}$ are the $M$ GSVs of $\boldsymbol{H}_2$ and $\boldsymbol{H}_1$ as shown in Appendix \ref{app:equif}.
	
	Using the intermediate results from above, we obtain
	\begin{align}
		\scalebox{0.8}{\mbox{\ensuremath{\displaystyle \boldsymbol{Z}}}} &= \scalebox{0.8}{\mbox{\ensuremath{\displaystyle \begin{bmatrix} \begin{bmatrix}\boldsymbol{K} & \bar{\boldsymbol{H}}_2\end{bmatrix}\hat{\boldsymbol{V}}_1\boldsymbol{\Sigma}_1^+\boldsymbol{Q}\begin{bmatrix}\hat{\boldsymbol{V}}_3 & \boldsymbol{0} \\ \boldsymbol{0} & \boldsymbol{I}_{\bar{M}_1}\end{bmatrix}\begin{bmatrix} 
		\left(\boldsymbol{I}_{M} + \boldsymbol{\Sigma}\right)^{-\frac{1}{2}} & \boldsymbol{0}\\ \boldsymbol{0} & \boldsymbol{I}_{\bar{M}_1}
		\end{bmatrix} & \frac{1}{\sqrt{\mathstrut\bar{M}_2}}\bar{\boldsymbol{H}}_1\hat{\boldsymbol{V}}_2\end{bmatrix},}}} \label{eqn:x}\\
		\scalebox{1}{\mbox{\ensuremath{\displaystyle \boldsymbol{Q}_1}}} &\scalebox{1}{\mbox{\ensuremath{\displaystyle = \begin{bmatrix}\hat{\boldsymbol{V}}_3^\mathrm{H} & \boldsymbol{0} \\ \boldsymbol{0} & \boldsymbol{I}_{\bar{M}_1}\end{bmatrix} \boldsymbol{Q}^\mathrm{H} \hat{\boldsymbol{U}}_1^\mathrm{H},}}} \qquad 
		\boldsymbol{Q}_2 = \begin{bmatrix}\boldsymbol{I}_{\bar{M}_2} & \boldsymbol{0}\\ \boldsymbol{0} & \hat{\boldsymbol{U}}_3^\mathrm{H}\end{bmatrix}\hat{\boldsymbol{U}}_2^\mathrm{H}.
	\end{align}
	Substituting the above matrices into $\boldsymbol{Q}_1\boldsymbol{H}_1\boldsymbol{Z}$ and $\boldsymbol{Q}_2\boldsymbol{H}_2\boldsymbol{Z},$ we obtain
	\begin{align}
		\scalebox{0.8}{\mbox{\ensuremath{\displaystyle \boldsymbol{Q}_1\boldsymbol{H}_1\boldsymbol{Z} = \begin{bmatrix}
		\underbrace{\left(\boldsymbol{I}_{M} + \boldsymbol{\Sigma}\right)^{-\frac{1}{2}}}_{\coloneqq\boldsymbol{\Sigma}_1} & \boldsymbol{0} & \boldsymbol{0} \\
		\boldsymbol{0} & \underbrace{\boldsymbol{I}_{\bar{M}_1}}_{\coloneqq\boldsymbol{D}_1} & \boldsymbol{0}
		\end{bmatrix}, \qquad \boldsymbol{Q}_2\boldsymbol{H}_2\boldsymbol{Z} = \begin{bmatrix}
		\underbrace{\boldsymbol{A}_3 \hat{\boldsymbol{V}}_3^\mathrm{H}\left(\boldsymbol{I}_{M} + \boldsymbol{\Sigma}\right)^{-\frac{1}{2}}}_{\coloneqq\boldsymbol{T}} & \boldsymbol{0} & \underbrace{\frac{1}{\sqrt{\mathstrut\bar{M}_2}}\hat{\boldsymbol{\Sigma}}_2}_{\coloneqq\boldsymbol{D}_2} \\
		\underbrace{\boldsymbol{\Sigma}\left(\boldsymbol{I}_{M} + \boldsymbol{\Sigma}\right)^{-\frac{1}{2}}}_{\coloneqq\boldsymbol{\Sigma}_2} & \boldsymbol{0} & \boldsymbol{0}
		\end{bmatrix},}}} \label{eqn:qhz}
	\end{align}
	where $\boldsymbol{A}_3 \in  \mathbb{C}^{\bar{M}_2\times M}$ contains the first $\bar{M}_2$ rows and the first $M$ columns of $\tilde{\boldsymbol{H}}_2\boldsymbol{Q},$ and $\hat{\boldsymbol{\Sigma}}_2 \in \mathbb{R}^{\bar{M}_2\times \bar{M}_2}$ contains the first $\bar{M}_2$ rows of $\boldsymbol{\Sigma}_2,$ which leads to (\ref{eqn:mgsvd}). \qed

	\subsection{Proof of Proposition \ref{prop:equif}}
	\label{app:equif}
	In the following, we consider the case $M_1,M_2 \geq N.$ The proofs for the remaining cases specified in Table \ref{tab:map} follow analogously.
	
	From Appendix \ref{app:mgsvd}, we note that $\boldsymbol{\Sigma}$ contains the singular values of matrix $\boldsymbol{B}_3$ in (\ref{eqn:b3}). For $M_1,M_2 \geq N,$ we have $\bar{\boldsymbol{H}}_1 = \bar{\boldsymbol{H}}_2 = \{\}.$ As the non-zero singular values of a matrix are unaffected by multiplication with unitary matrices and the conjugate transpose operation, the non-zero singular values of $\boldsymbol{B}_3,$ denoted by $\sigma\left(\boldsymbol{B}_3\right),$ can be simplified to
	\begin{align}
		\sigma\left(\boldsymbol{B}_3\right) = \sigma\left(\boldsymbol{Q}\boldsymbol{R}\right) = \sigma\left(\tilde{\boldsymbol{H}}_2\right) \overset{(a)}{=} \sigma\left(\boldsymbol{H}_2\boldsymbol{K} \hat{\boldsymbol{V}}_1\boldsymbol{\Sigma}_1^+\right) \overset{(b)}{=} \sigma\left(\boldsymbol{H}_2\boldsymbol{H}_1^+\right),
	\end{align}
	where (a) is obtained from (\ref{eqn:th2}) and (b) follows by noting that $\hat{\boldsymbol{V}}_1\boldsymbol{\Sigma}_1^+\hat{\boldsymbol{U}}_1^\mathrm{H} = \boldsymbol{H}_1^+$ and $\boldsymbol{K} = \boldsymbol{I}_N.$ Hence, the non-zero singular values of $\boldsymbol{B}_3$ are the same as those of matrix $\boldsymbol{H}_2 \boldsymbol{H}_1^+,$ which are the GSVs of $\boldsymbol{H}_2$ and $\boldsymbol{H}_1$ as they are solutions $\mu$ to $\mathrm{det}\left(\mu^2 \boldsymbol{I}_N - \boldsymbol{H}_2 (\boldsymbol{H}_1^H \boldsymbol{H}_1)^{-1} \boldsymbol{H}_2\right) = 0,$ thereby proving the assertion in Appendix \ref{app:mgsvd}. Furthermore, the squares of the singular values of $\boldsymbol{B}_3$ have the same distribution as the eigenvalues of the F-distributed matrix \cite{Perlman1977, Forrester2014} $\boldsymbol{F} = (\boldsymbol{W}_2)^{\frac{1}{2}} \boldsymbol{W}_1^{-1} (\boldsymbol{W}_2)^{\frac{1}{2}},$ where $\boldsymbol{W}_2 \sim \mathcal{CW}_{N}(M_2,\boldsymbol{I}_{M})$ and $\boldsymbol{W}_1 \sim \mathcal{CW}_{N}(M_1, \boldsymbol{I}_{M})$ are independent Wishart distributed matrices \cite[Thm. 3.2.4, Thm. 3.3.10, and Thm. 3.4.2]{Gupta1999}, \cite[Sec. 2]{Paige1981}, which completes the proof for the case $M_1,M_2 \geq N.$
	
	Lastly, as seen from (\ref{eqn:qhz}), matrices $\boldsymbol{\Sigma}_1^2$ and $\boldsymbol{\Sigma}_2^2$ are constructed as
	$\boldsymbol{\Sigma}_1^2 = (\boldsymbol{I}_M + \boldsymbol{\Sigma})^{-1}$ and $\boldsymbol{\Sigma}_2^2 = \boldsymbol{\Sigma}(\boldsymbol{I}_M + \boldsymbol{\Sigma})^{-1},$ which yields (\ref{eqn:equif}).\qed

	\subsection{Proof of Theorem \ref{thm:f}}
	\label{app:f}
	The density of the matrix-variate F-distribution is given by \cite{Gupta1999, Perlman1977}
	\begin{equation}
	K_F \frac{\mathrm{det}\left(\boldsymbol{F}\right)^{m_2-q}}{\mathrm{det}\left(\boldsymbol{I}_q + \boldsymbol{F}\right)^{m_1+m_2}},
	\end{equation}
	where $K_F$ is a proportionality constant that ensures that the pdf integrates to one. By performing a change of variables to the eigenvalues and eigenvectors of $\boldsymbol{F}$, and integrating over the eigenvectors \cite[Proposition 1.3.4]{Forrester2010}, we obtain the joint distribution of the real-valued eigenvalues, analogous to the results in \cite{Perlman1977}, \cite[Introduction]{Forrester2014}, as  follows:
	\begin{align}
	K_{\boldsymbol{\lambda}}\prod_{l = 1}^{q} \frac{\lambda_l^{m_2-q}}{(1+\lambda_l)^{m_1+m_2}} \prod_{1 \leq j < k \leq q} (\lambda_j - \lambda_k)^2, \label{eqn:joint}
	\end{align}
	where $\lambda_1 \geq \lambda_2 \geq \dots \geq \lambda_q$ denote the eigenvalues, and $K_{\boldsymbol{\lambda}}$ is a constant ensuring that the integral over the joint eigenvalue pdf is equal to one. Next, in order to obtain the marginal eigenvalue pdf, using Definition \ref{def:mepdf}, we integrate out the $q-1$ largest eigenvalues. The $q-1$ dimensional integration is performed using $\phi_i(\lambda_j) = \lambda_j^{i-1},$ $\psi_i(\lambda_j) = \lambda_j^{i-1},$ and $\xi(\lambda) = \dfrac{\lambda^{m_2-q}}{(1+\lambda)^{m_1+m_2}}$ in \cite[Corollary 1, Theorem 1]{Zanella2009} to obtain the expression in (\ref{eqn:f}). \qed

	\subsection{Proof of Theorem \ref{thm:f2}}
	\label{app:f2}
	The joint eigenvalue pdf of an F-distributed matrix $\boldsymbol{F}$ is given in (\ref{eqn:joint}) in Appendix \ref{app:f}. Next, starting from Definition \ref{def:oepdf}, the $l$-th marginal eigenvalue pdf is derived as follows:
	\begin{align}
	\scalebox{0.8}{\mbox{\ensuremath{\displaystyle p_l(\lambda_l; m_1, m_2, q)}}} &\scalebox{0.8}{\mbox{\ensuremath{\displaystyle = 
		\int_{\lambda_l}^{+\infty}\dots\int_{\lambda_2}^{+\infty} \left[\int_{0}^{\lambda_l}\dots\int_{0}^{\lambda_{q-1}}  p_{\boldsymbol{\lambda}}(\boldsymbol{\lambda}) \mathrm{d}\lambda_q\cdots\mathrm{d}\lambda_{l+1}\right] \mathrm{d}\lambda_1\cdots\mathrm{d}\lambda_{l-1}}}} \nonumber \\
	&\scalebox{0.8}{\mbox{\ensuremath{\displaystyle \overset{(a)}{=} K_{p_l} \sum_{n_1=1}^{q}\sum_{n_2\neq n_1=1}^{q}\dots\sum_{n_l\neq n_1,\dots,n_{l-1}=1}^{q} \sum_{m_1=1}^{q}\sum_{m_2\neq m_1=1}^{q}\dots\sum_{m_l\neq m_1,\dots,m_{l-1}=1}^{q}  s((n_1,\dots,n_l),(m_1,\dots,m_l))}}} \nonumber\\
	& \scalebox{0.8}{\mbox{\ensuremath{\displaystyle \quad\times \mathrm{det}\left(\boldsymbol{\Xi}\left(l, m_1, m_2, q, {\mathcal{I}}^{[l,(n_1,\dots,n_l)]},{\mathcal{I}}^{[l,(m_1,\dots,m_l)]}\right)\right) \varphi(n_l,m_l,\lambda_l)\prod_{i = 1}^{l-1} \int_{\lambda_l}^{+\infty} \varphi(n_i,m_i,\lambda) \mathrm{d} \lambda}}}, \label{eqn:pr1} 
	\end{align}
	where $K_{p_l}$ is a constant ensuring that the integral over the pdf equals one, and $(a)$ is obtained using \cite[Eq. (50)]{Zanella2009}, which can be used for computing the integral in (\ref{eqn:defoe}). Based on (\ref{eqn:joint}), the function $\varphi(n,m,x)$ in \cite[Eq. (50)]{Zanella2009} is written as $\varphi(n,m,x) = \phi_n(x)\psi_m(x)\xi(x),$ where $\phi_n(x) = x^{n-1},$ $\psi_m(x) = x^{m-1},$ and $\xi(x) = \dfrac{x^{m_2-q}}{(1+x)^{m_1+m_2}}.$ Furthermore, the multiple summations in (\ref{eqn:pr1}) can be equivalently reformulated in terms of the recurrence relation in (\ref{eqn:rec}). Lastly, special integrals are computed using the definition of the incomplete Beta function and the hypergeometric function from \cite{Olver2010} to obtain the expressions in Theorem \ref{thm:f2}. \qed

	\end{appendices}
	\bibliographystyle{IEEEtran}
	\bibliography{references}
	
\end{document}